\documentclass[12pt]{amsart}

\usepackage{cancel}
\usepackage{mathtools,dsfont}
\usepackage[english]{babel}
\usepackage[hmargin=0.8in,height=8.6in]{geometry}
\usepackage{amssymb,amsthm, times}
\usepackage{delarray,verbatim}
\usepackage{ifpdf}
\usepackage{cases}
\ifpdf
\usepackage[pdftex]{graphicx}
 \else
\usepackage[dvips]{graphicx}
 \fi

\usepackage{hyperref}
\usepackage{enumitem}
\usepackage{bm}

\usepackage{ifpdf}
\usepackage{color}
\definecolor{webgreen}{rgb}{0,.5,0}
\definecolor{webbrown}{rgb}{.8,0,0}
\definecolor{emphcolor}{rgb}{0.5,0.95,0.95}

\usepackage{hyperref}
\hypersetup{%
	colorlinks=true,
	linkcolor=webbrown,
	filecolor=webbrown,
	citecolor=webgreen,
	breaklinks=true}
\ifpdf \hypersetup{pdftex,
	pdfstartview=FitH, 
	bookmarksopen=true,
	bookmarksnumbered=true
} \else \hypersetup{dvips} \fi
\allowdisplaybreaks

\mathtoolsset{showonlyrefs}

\numberwithin{equation}{section}

\newtheorem{theorem}{Theorem}[section]
\newtheorem{proposition}[theorem]{Proposition}
\newtheorem{corollary}[theorem]{Corollary}
\newtheorem{remark}[theorem]{Remark}
\newtheorem{lemma}[theorem]{Lemma}

\newtheorem{assump}[theorem]{Assumption}
\newtheorem{definition}[theorem]{Definition}

\newcommand {\R}{\mathbb{R}}

\newcommand {\F}{\mathcal{F}}

\newcommand {\p}{\mathbb{P}}

\newcommand {\E}{\mathbb{E}}

\newcommand{\diff}{{\rm d}}

\newcommand{\e}{\mathbb{E}}

\newcommand{\Gen}{\mathcal{L}}
\newcommand{\cC}{\mathcal{C}}

\newcommand{\Ind}{\mathds{1}}
\newcommand{\ind}{\mathds{1}}

\newcommand{\cL}{\mathcal{L}}

\newcommand{\admissible}{\Pi}

\newcommand{\BRA}[1]{{{\left\{#1\right\}}}} 
\newcommand{\PAR}[1]{{{\left(#1\right)}}} 
\newcommand{\SBRA}[1]{{{\left[#1\right]}}} 
\renewcommand{\leq}{\leqslant}
\renewcommand{\geq}{\geqslant}

\begin{document}

\title[Optimization of capital injections and AC dividend payments]{Optimization of capital injections and absolutely continuous dividend payments in a diffusion model}

\author[H. Guérin]{H\'el\`ene Gu\'erin}
\address{Département de mathématiques, Université du Québec à Montréal (UQAM)}
\email{guerin.helene@uqam.ca}

\author[D. Mata]{Dante Mata}
\address{D\'epartement de math\'ematiques, Universit\'e du Qu\'ebec \`a Montr\'eal (UQAM)}
\email{mata\_lopez.dante@uqam.ca}

\author[J.-F. Renaud]{Jean-Fran\c{c}ois Renaud}
\address{D\'epartement de math\'ematiques, Universit\'e du Qu\'ebec \`a Montr\'eal (UQAM)}
\email{renaud.jf@uqam.ca}

\author[A. Roch]{Alexandre Roch}
\address{D\'epartement de finance, \'Ecole des sciences de la gestion, Universit\'e du Qu\'ebec \`a Montr\'eal (UQAM)}
\email{roch.alexandre\_f@uqam.ca}

\date{}

\thanks{Funding in support of this work was provided by a CRM-ISM Postdoctoral Fellowship, a FRQNT Postdoctoral Fellowship (369449) and Discovery Grants (RGPIN-2025-05758, RGPIN-2020-07239) from the Natural Sciences and Engineering Research Council of Canada (NSERC)}

\thanks{This work was supported by the Research Institute for Mathematical Sciences, an International Joint Usage/Research Center located in Kyoto University.}

\begin{abstract}
We investigate a joint optimization problem of dividend payments and capital injections for a surplus process driven by a general diffusion. Dividend payments are assumed to be absolutely continuous in time, with the dividend rate bounded by a nonnegative concave function of the current surplus; while capital injections are modelled by a general nondecreasing process. 
We first analyze an auxiliary bail-out problem in which capital injections are required to keep the surplus nonnegative at all times. Under a concavity assumption on the drift, we prove that the associated value function is concave and is a classical solution of the corresponding Hamilton--Jacobi--Bellman (HJB) equation. We further characterize an optimal policy as a refraction--reflection strategy: the surplus is 
reflected at zero by capital injections, while dividends are paid at the maximal admissible rate whenever the surplus exceeds an optimal threshold.
Our main contribution establishes that the general optimization problem exhibits a Løkka--Zervos-type dichotomy. More precisely, an optimal policy is either a dividend refraction strategy without injections, in which ruin occurs, or a refraction--reflection dividend--injection strategy. This optimal injection decision is characterized through a simple comparison of the auxiliary value functions at the origin.
		\\
		\noindent \small{\noindent  AMS 2020 Subject Classifications: 93E20, 60J60, 60J70.}
	\end{abstract}
	
	\maketitle


\section{Introduction}
We study an optimal dividend and capital injection problem for a surplus process driven by a general diffusion, in the spirit of, e.g., \cite{sethi-taksar_2002} and \cite{zhu-yang_2016}. We assume that the cumulative dividend process is absolutely continuous in time and that the dividend rate is subject to a surplus-dependent upper bound. More precisely, following \cite{guerin-et-al_2024}, the admissible dividend rate is bounded by a nonnegative, nondecreasing, and concave function of the current surplus level. Capital injections, on the other hand, are modelled, as in \cite{RRS2026, sethi-taksar_2002, zhu-yang_2016}, by a nondecreasing, nonnegative, càdlàg adapted process.

A well-known drawback of the classical dividend problem is that ruin occurs almost surely for a wide class of models when no capital injections are allowed. Yet insolvency may result from adverse shocks even when the firm remains fundamentally profitable, so that shareholders may rationally choose to recapitalize the firm by injecting additional funds whenever the expected benefits outweigh the associated costs.
This observation motivates the introduction of capital injections and has generated an extensive literature. In diffusion models, optimal injection policies often take the form of a lower reflection: the surplus is reflected (in the Skorohod sense) at a lower boundary whenever it attempts to cross it from above. For instance, Guo and Pham \cite{guo-pham_2005} establish the optimality of a double-barrier reflection strategy for a model driven by geometric Brownian motion, while Kulenko and Schmidli \cite{KS2008} prove an analogous result in the Cramér–Lundberg setting. More recently, Gajek and Kucinski \cite{gajek-kucinski_2017} considered a spectrally negative Lévy process. These works, however, typically treat dividend payments as singular controls, whereas the present paper focuses on absolutely continuous dividend strategies under a level-dependent bound, extending the results of Renaud et al. \cite{RRS2026} to a general diffusion framework.

The above considerations highlight a fundamental trade-off between payouts and injections. Such a trade-off is not specific to financial settings and also arises in other areas of applied probability. For instance, in the natural resource management literature, excessive harvesting may drive a population to extinction, while costly interventions such as seeding can be used to sustain it (see, e.g., \cite{Alvarez1998, Lande1994, Lande1995}). In that context, the objective is to determine optimal harvesting–seeding strategies, which balance immediate gains against the long-term viability of the system. Related stochastic control problems have been studied in diffusion models in this context (see, e.g., \cite{Hening2019, Hening2020}), further highlighting the generality of this class of problems. However, models from the resource management literature typically impose more model-specific assumptions on the diffusion coefficients.

Unlike a large strand of the dividend literature, which imposes capital injections to prevent ruin and reflect the surplus at zero (the so-called bail-out problem; see, e.g., \cite{avram-et-al_2007,FERRARI2019,FS2019,morenofranco2024,NPY2020}), we treat recapitalization as an endogenous decision. We therefore formulate a joint dividend and capital-injection control problem up to ruin, in which the controller may optimally decide to collect dividends, inject capital at any given time and any chosen level of surplus, or allow the surplus to eventually reach ruin. The class of admissible strategies is broad and includes, in particular, both singular and absolutely continuous capital injections. Despite this flexibility, one of our main contributions is to show that the optimization problem reduces to a simple dichotomy. More precisely, the value of the general problem coincides with that of one of the following two auxiliary problems:
\begin{itemize}
\item A dividend problem in which injections are not possible, so that dividends are paid until the surplus hits zero;
\item A bail-out dividend problem in which capital injections are mandatory in order to keep the surplus process nonnegative.
\end{itemize}

This kind of dichotomic result was first obtained by Løkka and Zervos in \cite{lokka-zervos08} and later extended to several more general settings. For instance, Avanzi et al. \cite{AvanziShenWong2011} consider a Brownian motion with drift and negative jumps, emphasizing the role of recapitalization in preventing inefficient ruin, while Lindensj\"o and Lindskog \cite{lindensjo2020} study singular dividend controls for a general diffusion. 
Similarly, in the case of Cramér--Lundberg models, Avram et al. \cite{AvramGoreacRenaud2019} obtain a Løkka--Zervos-type alternative for exponential jumps, while Avram et al. \cite{AvramGoreacLiWu2021} study a variant in which capital injections are allowed only for small claims, whereas larger claims cause ruin.
More recently, Roch \cite{Roch2025} investigates a capital injection and dividend problem for a general Lévy process with constant coefficients.

A key feature of the present work is that dividends are restricted to be absolutely continuous and subject to a state-dependent concave upper bound. This setting is directly connected to recent de Finetti-type problems in the Brownian framework involving bounded dividend rates \cite{locas-renaud_2024} and linearly bounded rates \cite{RS2021}. An optimization dichotomy in our setting was also established for a  Brownian model in \cite{RRS2026}. We extend these results to a general diffusion framework, where the loss of spatial homogeneity gives rise to substantial additional analytical difficulties.

The above dichotomy reduces the broad class of admissible strategies to two auxiliary problems, depending on whether capital injections are excluded or forced. In our setting, the former was studied in \cite{guerin-et-al_2024}, where an optimal policy is characterized by a refraction strategy. Therefore, the first objective of the present work is to analyze the value function of the forced-injection problem. The problem is more involved in our setting than in the constant-coefficient framework, where concavity and the optimality of bailing out the surplus process only at the origin are established by direct probabilistic arguments. In that simplified setting, Renaud et al. \cite{RRS2026} in fact prove the optimality of injections only at the origin on a pathwise basis. Such an argument greatly simplifies the treatment of capital injections and the construction of a candidate value function. However, it is no longer available in the more general diffusion setting considered here. We therefore establish these two key properties through an analytical study of the HJB equation. In particular, in Section~\ref{sec:optimality-reflection-refraction}, we show that an optimal control consists of (i) a Skorokhod-type reflection at zero, corresponding to instantaneous injections that keep the surplus nonnegative,
 and (ii) a refraction at an optimal threshold $b_c$, above which dividends are paid at the maximal admissible rate. This characterization follows from Corollary~\ref{cor:a-star-zero} and Proposition~\ref{prop:explicit-Vc}.
The presence of capital injections makes the proof of optimality substantially more involved than in \cite{guerin-et-al_2024}. In that work, an explicit
value function allows a direct guess-and-verify argument. Here, by contrast, we analyze the associated HJB equation using viscosity methods in order to derive the key analytic properties of the value function. Concavity of the value function is obtained by proving that its concave envelope is a subsolution of the HJB equation  and invoking a comparison principle. This yields the optimality of capital injections only at the origin as a final consequence. Our approach, based on viscosity solution techniques for HJB equations, is in line with recent developments in the stochastic control literature; see, e.g., \cite{AAM2022,Guan-Xu_2024,AAM2026}.
We further prove that the value function is a classical $\mathcal{C}^2$ solution of the HJB equation, characterize the optimal dividend barrier through a smooth-fit argument, and derive a representation in terms of the fundamental solutions of the associated differential equation.

Finally, in Section~\ref{sec:Optimal:2}, a comparison principle for the value functions of the two auxiliary problems in the dichotomy implies that our optimal injection decision is determined by a boundary test comparing the value functions (or their derivatives) at zero, indicating whether injections are optimal. Our main result, Theorem~\ref{thm.dichotomy}, shows that an optimal strategy reduces to the following alternative for the controller:
\begin{enumerate}
\item Pay dividends and never inject capital according to a pure refraction strategy, as studied in \cite{guerin-et-al_2024}, so that ruin occurs at the first passage time below zero;
\item Pay dividends and inject capital at the origin according to a refraction-reflection strategy, meaning that ruin never occurs.
\end{enumerate}
We now introduce the model more precisely.

\subsection{Model and problem formulation}\label{sect:model-problem}

Let $(\Omega, \F, (\F_t)_{t\geq 0}, \p)$ be a filtered probability space satisfying the \textit{usual conditions}. We assume the uncontrolled surplus process $Y=(Y_t)_{t\geq 0}$ evolves according to the following stochastic differential equation: for $t \geq 0$,
\begin{equation}\label{eq:uncontrolled-SDE}
\diff Y_t = \mu(Y_t) \diff t + \sigma (Y_t) \diff W_t, 
\end{equation}
where $W=(W_t)_{t\geq 0}$ is a standard Brownian motion (adapted to the filtration) and where the coefficients are assumed to satisfy the following conditions.

Following \cite{sethi-taksar_2002}, we assume that the drift function is concave, an assumption that is carefully motivated and economically interpreted in this setting. This structural condition enables a clear characterization of an optimal policy and yields a complete solution to our control problem.  

\begin{assump}\label{assumption:SDE}
Fix a time-preference (discounting) parameter $q>0$.
     \begin{enumerate}
         \item The drift coefficient $\mu \colon [0,\infty) \to \R$ is a differentiable concave function such that $\mu'(x) < q$ for all $x \in [0,\infty)$.
         \item The volatility coefficient $\sigma \colon [0,\infty) \to [0,\infty)$ is a Lipschitz function such that $\sigma (x) > \underline{\sigma}$ for all $x \in [0,\infty)$, where $\underline{\sigma}>0$ is a constant.  
     \end{enumerate}
 \end{assump}

In the above, and for the rest of the paper, for a \textit{sufficiently smooth} function $f \colon [0,\infty) \to \R$, we write $f^\prime(0):=f^\prime(0+)$. 

Note that conditions (1) and (2) are sufficient to guarantee the existence of a unique strong solution to the stochastic differential equation~\eqref{eq:uncontrolled-SDE}; see, e.g., Theorem IV.3.1 in \cite{ikeda-watanabe_1989}. The growth restriction assumption on $\mu$, within condition (1), has appeared before (and been discussed) in the literature; see, e.g., \cite{paulsen_2007,ekstrom-lindensjo_2023}. In our setting, this growth condition guarantees  the existence of non-trivial solutions and the finiteness of important mathematical quantities. In addition, several key steps in our proofs rely on the concavity of $\mu$, also within condition (1) of Assumption~\ref{assumption:SDE}. 

In our optimization problem, a dividend-capital injection strategy is a pair $\pi = (\ell, K)$ of adapted and nonnegative stochastic processes $\ell=(\ell_t)_{t\geq 0}$ and $K=(K_t)_{t\geq 0}$ where $K$ has nondecreasing and \textit{càdlàg} trajectories and is such that $K_{0-}=0$. For a fixed time $t \geq 0$, $K_t$ represents the cumulative amount of capital injections made up to (and including) time $t$, while $L_t := \int_0^t \ell_s \diff s$ represents the cumulative amount of dividends made up to time $t$. Thus, the process $\ell$ represents the rate at which dividends are paid out. Note that in our setup, the stochastic process $L=(L_t)_{t\geq 0}$ is absolutely continuous (with respect to the Lebesgue measure) with nondecreasing trajectories and it is such that $L_0=0$.

For a strategy $\pi = (\ell, K)$, the dynamic of its associated controlled surplus process $X^\pi$ is given by: for $t\geq 0$,
\begin{equation}\label{SDE_K}
\diff X^\pi_t = (\mu(X^\pi_t) - \ell_t)\diff t + \sigma(X^\pi_t)\diff W_t + \diff K_t . 
\end{equation}

To define our set of admissible strategies, let us fix a differentiable function $F \colon [0,\infty) \to [0,\infty)$ that is concave, nondecreasing and such that $F(0) \geq 0$. A strategy $\pi = (\ell, K)$ is \textit{admissible} if:
\begin{enumerate}
    \item $0 \leq \ell_t \leq F(X^\pi_t)$ for all $t \geq 0$;
    \item\label{cond-admissibility-on-K} $\E_x\SBRA{\int_{[0,\infty)} e^{-qt} \diff K_t} < \infty$.
\end{enumerate}
Let $\admissible$ be the set of all admissible strategies. Note that the assumptions on $F$ have also appeared before in the literature; see \cite{locas-renaud_2024,RRS2026, guerin-et-al_2024}. 

Let us introduce two subsets of admissible strategies:
\[
\Pi_d:=\BRA{(\ell,K)\in\Pi \colon K\equiv 0}
\]
is the subset of strategies in which injections are excluded, and
\[
\Pi_c:=\BRA{\pi=(\ell,K)\in\Pi \colon X^{\pi}_t \geq 0 \text{ for all }t \geq 0}
\]
is the subset of strategies that include enough capital injections to prevent bankruptcy at all times.

Let us fix $\beta > 1$, the cost per unit of capital injected. The performance function associated to an admissible strategy $\pi=(\ell,K) \in \admissible$ is defined as 
\begin{equation}\label{eq:performance-function}
J(x;\pi) := \E_x\SBRA{ \int_0^{\tau_0^\pi} e^{-qt} \ell_t \diff t - \beta \int_{[0,\tau_0^\pi)} e^{-qt} \diff K_t }, \quad x \geq 0 ,
\end{equation}
in which $X^\pi_{0-}=x$ and $\tau_0^\pi=\inf\BRA{t\geq 0:X_t^\pi< 0}$. The ruin time $\tau_0^\pi$ is set to $\infty$ if $X^\pi$ stays nonnegative. The objectives of such an optimization problem is to compute its value function, which is given here by
\begin{equation}\label{value:dichotomy}
V(x) := \sup_{\pi \in \Pi} J(x;\pi), \quad x \geq 0 ,
\end{equation}
and, if one exists, to obtain an optimal admissible strategy.

\begin{remark}\label{remark:pi-c:pi-d}
As $\Pi$ includes strategies that belong neither to $\Pi_c$ nor to $\Pi_d$, the set $\Pi_d$ is strictly included in $\Pi\setminus\Pi_c$. 
Indeed, the set of admissible strategies $\Pi \setminus \Pi_c$ contains strategies under which capital injections may occur before ruin is eventually declared. Moreover, note that $\Pi_c$ contains strategies avoiding ruin, which can be done in several different ways. For example, one could inject whenever the process reaches a pre-determined positive level (that is deemed too close to zero), either through lump-sum payments or by reflecting the process at this positive level.
Also, if $\pi \in \Pi_c$, then $\tau_0^\pi=\infty$ almost surely, while if $\pi \in \Pi \setminus \Pi_c$, then the event $\BRA{\tau_0^\pi<\infty}$ has a positive probability.  
\end{remark}

We also define the value functions corresponding to the subsets of admissible strategies $\Pi_c$ and $\Pi_d$: for $x \geq 0$,
\[
V_d(x) := \sup_{\pi \in \Pi_d} J(x;\pi) \quad \text{and} \quad V_c(x) := \sup_{\pi \in \Pi_c} J(x;\pi) .
\]
Of course, by definition, we have
\[
V(x) \geq \max\BRA{V_d(x),V_c(x)}, \quad x \geq 0.
\]
Note that the value function $V_d$ has been computed and studied in \cite{guerin-et-al_2024} under the assumption that $\mu(0)>0$; we will come back to this function (and this extra assumption) later.

\subsection{Organization of the paper}

The rest of this paper is organized as follows. In Section~\ref{Sec:HJB}, we study $V_c$ from an analytical point of view, culminating with it being the unique classical solution to a free-boundary value problem. In Section~\ref{Sec:Optimal:1}, we identify an optimal strategy in $\Pi_c$, that is a strategy yielding the optimal value $V_c$. Finally, in Section~\ref{sec:Optimal:2}, we solve our main joint optimization problem for dividend payments and capital injections, using a comparison procedure, and we obtain the announced dichotomy.

\section{Analytical properties of the value function $V_c$}\label{Sec:HJB}

In this section, we solve the control problem corresponding to strategies in $\Pi_c$, i.e., strategies for which ruin must be avoided. More precisely, we will prove that the value function $V_c$ is concave and twice continuously differentiable, and that it is the performance function of a specific \textit{refraction-reflection strategy}. In particular, we will prove that an optimal capital injection strategy is a Skorohod reflection at level 0. Our proof is analytical rather than probabilistic (see \cite{RRS2026}) as we are working within a general diffusion framework.

First, note that the value function $V_c$ is finite. Indeed, we will prove in Section~\ref{Sec:Optimal:1} that $V_c(x)>-\infty$ for every $x \geq 0$ by exhibiting a strategy whose performance is greater than $-\infty$ (see Remark~\ref{rk:V_c finite}), independently  of the results in the rest of that section. Moreover, the next proposition readily implies that $V_c(x)<\infty$.
\begin{proposition}\label{prop:immediate-injection}
For all $x \geq 0$, we have
\[
V_c(x) \leq x+\frac{\mu(0)_+}{q} ,
\]
where $\mu(0)_+=\max(\mu(0),0)$.
Also, for all $x \geq 0$ and $h\geq 0$, we have
\begin{equation}\label{ineq.opt}
V_c(x+h) - \beta h \leq V_c(x) \leq V_c(x+h).
\end{equation}
In particular, $V_c$ is Lipschitz continuous.
\end{proposition}

\begin{proof}
Take $\pi=(\ell,K)\in\Pi_c$. Recall that, for such a strategy, we have $X^\pi_t \geq 0$ for all $t \geq 0$. By Itô's Formula and concavity of $\mu$, we can write
\begin{align*}
&\E_x\SBRA{\int_0^t\ell_se^{-qs}\diff s-\beta\int_{[0,t]}e^{-qs}\diff K_s}\\
&=x-\E[e^{-qt}X_t^\pi]+\int_0^t\E\SBRA{e^{-qs}\PAR{\mu(X^\pi_s)-qX^\pi_s}}\diff s+(1-\beta)\E\SBRA{\int_{[0,t]} e^{-qs}\diff K_s}\\
&\leq x-\E[e^{-qt}X_t^\pi]+\frac{\mu(0)_+}{q}+(\mu'(0)-q)\int_0^t\E[e^{-qs}X_s^\pi]\diff s\\
&\leq x+\frac{\mu(0)_+}{q}.
\end{align*}
Note that we used the fact that $\beta>1$ and that $\mu'(0)<q$. Taking the limit when $t\to \infty$ gives the first result.

The first inequality in~\eqref{ineq.opt} simply follows from the fact that an immediate capital injection (at time $0$) is admissible but may not be optimal. The second inequality follows from the fact that any admissible strategy $\pi$ such that $X^\pi$ is kept nonnegative when started at $x$ will also keep $X^\pi$ nonnegative when started at $x+h$.
\end{proof}

It is easy to verify that, if $F$ is constant, for instance if $F(x)=\ell_\infty$ for all $x \geq 0$, then
\begin{equation}\label{ineq.constant-bound}
V_c(x) \leq \frac{\ell_\infty}{q} , \; \text{for all $x \geq 0$.}
\end{equation}

\subsection{Dynamic programming principle and HJB equation}

Next, we state the Dynamic Programming Principle (DPP) associated to the control problem under study. To do so, let us introduce the following notation: for a given a stopping time $\tau$, define
\[
\int_{[0,\tau]} e^{-qt} \diff K_t := \int_{[0,\infty)} e^{-qt} \diff K_t , \quad \text{on the event $\{\tau=\infty\}$.}
\]

\begin{proposition}[Dynamic Programming Principle]\label{prop:dynamic programming principle}
For all $x\geq 0$, we have
\begin{equation}\label{DPP}
V_c(x) = \sup_{\pi=(\ell,K) \in \Pi_c} \E_x\SBRA{ \int_0^\tau e^{-qt} \ell_t \diff t - \beta \int_{[0,\tau]} e^{-qt} \diff K_t + e^{-q\tau} V_c(X^\pi_\tau) \Ind_{\tau<\infty}} ,
\end{equation}
where $\tau$ is a stopping time, possibly depending on $\pi$.
\end{proposition}

Since the value function $V_c$ is (Lipschitz) continuous, the proof of the above DPP can follow the same lines of reasoning as in the proof of Proposition~3.4 and Lemma~A.2 in \cite{NPY2020}. We also refer the reader to \cite{Bouchard-Touzi_2011} for a proof in a more general setting.

The HJB equation associated with our general control problem (with $V$) but also with the control problem studied in this section (with $V_c$) is as follows: for $x \in (0,\infty)$,
\begin{equation}\label{eq.HJB}
\min\BRA{q v(x) - \sup_{0 \leq \ell \leq F(x)} \mathcal L_\ell [v](x) , \; \beta - v'(x)} = 0 , 
\end{equation}
in which, for $v \in \mathcal{C}^2(0,\infty)$ and a constant $\ell \geq 0$,
\begin{equation}\label{eq.generatorL}
\mathcal L_\ell [v](x) = \frac{\sigma^2(x)}{2} v^{\prime \prime}(x) + (\mu(x)- \ell) v^\prime(x) + \ell .
\end{equation}
In the sequel, we will write $\cL$ for $\cL_0$, i.e., when $\ell=0$.

In this paper, the HJB equation~\eqref{eq.HJB} will appear with different boundary conditions, including the following one:
\begin{equation}\label{eq.boundary_Vc}
v^\prime(0) = \beta .
\end{equation}
Recall that we use the following notation: $v^\prime(0) \equiv v^\prime(0+)$.

One of our goals in this section is to prove that $V_c$ is a \textit{solution} of this HJB equation, and that it satisfies the boundary condition~\eqref{eq.boundary_Vc}. Let us start by proving the latter.

\begin{proposition} \label{prop.Vcprime}
The value function $V_c$ is such that $V_c^\prime (0) = \beta$.
\end{proposition}

A proof of this last proposition is provided in Appendix~\ref{Appendix:proof:Vcprime}. While a similar result has been proven in \cite{RRS2026} for a Brownian motion with drift using probabilistic tools, we remark that our proof is analytical and relies on the properties of $\mu, \sigma,$ and $F$.

\subsection{Classical supersolutions and subsolutions}

Next, let us accumulate results toward proving that $V_c$ is a \textit{solution} of the HJB equation~\eqref{eq.HJB}.

\begin{definition}\label{def:classical-super-subsolution}
A function $\psi \in \mathcal{C}^2 (0,\infty)$ is a classical supersolution (resp.\ subsolution) of the HJB equation~\eqref{eq.HJB} if for all $x \in (0,\infty)$, we have
\begin{equation}\label{eq:HJB-supersol}
\min \BRA{q \psi(x) - \sup_{0 \leq \ell \leq F(x)} \mathcal L_\ell [\psi] (x), \; \beta - \psi'(x)} \geq 0 , \mbox{ (respectively, $\leq 0$)} .
\end{equation}
The function $\psi$ is said to be a classical solution of the HJB equation~\eqref{eq.HJB} if it is both a classical supersolution and a classical subsolution.
\end{definition}

The following two lemmas provide comparison results for classical supersolutions and subsolutions of the HJB equation and the value function $V_c$. They are key to obtaining the analytical properties of $V_c$ and will also be useful in our analysis of $V$.

\begin{lemma}\label{lemma.super_comp}
If $\psi$ is a nonnegative classical supersolution of the HJB equation~\eqref{eq.HJB} and if $\psi(0) \geq V_c(0)$, then $\psi(x) \geq V_c(x)$ for all $x \in [0,\infty)$. 
\end{lemma}

\begin{lemma}\label{lemma.sub_comp}
Let $0 \leq a < b \leq \infty$ and let $B \in (0,\beta)$. If $\phi(x) = A+Bx$ is a classical subsolution on $(a,b)$ of the HJB equation~\eqref{eq.HJB}, if $\phi(a) \leq V_c(a)$ and if $\phi(b) \leq V_c(b)$ (when $b<\infty$), then $\phi(x) \leq V_c(x)$ for all $x \in [a,b]$.
\end{lemma}

Proofs of these last two lemmas are provided in Appendix~\ref{App.super_comp} and Appendix~\ref{App.sub_comp}, respectively.

From Lemma~\ref{lemma.super_comp}, we can deduce that the value function $V_c$ has at most linear growth. In this direction, let us define, for each $\xi \geq 0$,
\begin{equation}\label{def-of-B}
B(\xi) = \frac{F'(\xi)}{q - \mu'(\xi) + F'(\xi)}
\end{equation}
and 
\[
A(\xi) = \max \BRA{ 0, \; V_c(0) , \; \frac{B(\xi) \mu(\xi) + (1-B(\xi)) F(\xi)}{q} } .
\]
Note that, under Assumption~\ref{assumption:SDE}, we have $B(\xi) \in [0,1)$. 

\begin{proposition}\label{prop.affine_bound}
For any $\xi \geq 0$, we have
\[
V_c(x) \leq A(\xi) + B(\xi) x \; , \quad \text{for all $x \geq 0$.}
\]
\end{proposition}

\begin{proof} 
Fix $\xi \geq 0$ and define $\psi(x) = A(\xi) + B(\xi) x$. By definition, we have $B(\xi) \in [0,1)$ and $\psi(0)=A(\xi)\geq \max\BRA{0, V_c(0)}$. Consequently, $\psi$ is nonnegative and, by Lemma~\ref{lemma.super_comp}, all is left to prove is that $\psi$ is a classical supersolution of the HJB equation~\eqref{eq.HJB}.

First, since $B(\xi) \in [0,1)$, then it is clear that $\psi'(x) < \beta$ for all $x>0$. Second, we have
\begin{align*}
q \psi(x) - \sup_{0 \leq \ell \leq F(x)} \mathcal L_\ell [\psi] (x) &= q \psi(x) - \sup_{0 \leq \ell \leq F(x)} \PAR{\mu(x) B(\xi) + \ell (1-B(\xi))}\\
&= B(\xi) \PAR{qx-\mu(x)} - (1-B(\xi)) F(x) + q A(\xi) .
\end{align*}

By concavity of both $\mu$ and $F$, we have $\mu(x) \leq \mu(\xi) + \mu^\prime(\xi) (x-\xi)$ and $F(x) \leq F(\xi) + F^\prime(\xi) (x-\xi)$. As a consequence, we further have that
\begin{align*}
&q \psi(x) - \sup_{0 \leq \ell \leq F(x)} \mathcal L_\ell [\psi] (x)\\
&\geq B(\xi) \PAR{qx - \mu(\xi) - \mu^\prime(\xi) (x-\xi)} - (1-B(\xi)) \PAR{F(\xi) + F^\prime(\xi) (x-\xi)} + q A(\xi) \\
&= \SBRA{ B(\xi) \PAR{q - \mu^\prime(\xi)} - (1-B(\xi)) F^\prime(\xi) } (x-\xi) + q \xi B(\xi) - \SBRA{ B(\xi) \mu(\xi) + (1-B(\xi)) F(\xi) } + q A(\xi) \\
&\geq q A(\xi) - \SBRA{ B(\xi) \mu(\xi) + (1-B(\xi)) F(\xi) } \\
&\geq 0 ,
\end{align*}
where in the penultimate inequality we used \eqref{def-of-B} and the fact that $q \xi B(\xi) \geq 0$, while in the last inequality we used the definition of $A(\xi)$.
\end{proof}

Note that, if $F\equiv\ell_\infty$, then, for all $\xi \geq 0$, we have $B(\xi)=0$ and 
\[
A(\xi) = \max \BRA{ 0, \; V_c(0) , \; \frac{\ell_\infty}{q} } .
\]
So, we have recovered the inequality in~\eqref{ineq.constant-bound}.

\begin{remark}\label{rk:B(infinity)}
Since $F$ is nondecreasing and since both $F'$ and $\mu'$ are nonincreasing (by the concavity  assumptions), we can easily verify that $\xi\mapsto B(\xi)$ is a nonincreasing function. In addition, as $B$ is nonnegative, we deduce that $\lim_{\xi \to \infty} B(\xi)$ is well defined and lies in $[0,1)$.
\end{remark}

\begin{remark}
If $\psi(x)=A+Bx$ with $1\leq B \leq \beta$, then using the concavity of $\mu$, we can show that
\[
q \psi(x) - \sup_{0 \leq \ell \leq F(x)} \mathcal L_\ell [\psi] (x) \geq q A + B \SBRA{ \left(q - \mu^\prime(0) \right) x - \mu(0)} .
\]
Under our standing assumptions, to get a nonnegative classical supersolution, it suffices to choose $A$ large enough, i.e.,
\[
A \geq \frac{B \mu(0)}{q} .
\]
In particular, $\psi(x)=A+\beta x$  is a nonnegative classical supersolution whenever $A \geq \max \BRA{0, V_c(0), \frac{\beta \mu(0)}{q}}$. In this case, $\psi(0)\geq V_c(0)$ and $V_c \leq \psi$ on $[0,\infty)$ by Lemma \ref{lemma.super_comp}.  In fact, $\psi$ is a solution of \eqref{eq.HJB} since $\beta - \psi'(x) = 0$ for all $x\geq0.$ Therefore, the HJB equation~\eqref{eq.HJB} does not have a unique solution.
\end{remark}

\subsection{Viscosity solutions}

At this stage of our journey toward proving that $V_c$ is a \textit{solution} of the HJB equation~\eqref{eq.HJB}, we do not know whether the value function $V_c$ is twice continuously differentiable or not. Therefore, we will first notice that it is a \textit{viscosity solution} to this HJB equation. 

\begin{definition}\label{def:viscosity-solution}
Let $V \colon [0,\infty) \to \R$ be locally bounded.
\begin{enumerate}
\item\label{testfunction-1} V is a viscosity subsolution of \eqref{eq.HJB} if for all $x \in (0,\infty)$, and for all $\phi \in C^2(0,\infty)$ such that $\phi \geq V$ and $\phi(x) = V(x)$, we have
\[
\min\BRA{q \phi(x) - \sup_{0 \leq \ell \leq F(x)} \mathcal L_\ell [\phi](x) , \beta - \phi'(x)}  \leq 0 .
\]
\item\label{testfunction-2} V is a viscosity supersolution of \eqref{eq.HJB} if for all $x \in (0,\infty)$, and for all $\phi \in C^2(0,\infty)$ such that $\phi \leq V$ and $\phi(x) = V(x)$, we have
\[
\min\BRA{q \phi(x) - \sup_{0 \leq \ell \leq F(x)} \mathcal L_\ell [\phi](x) , \beta - \phi'(x)}  \geq 0 .
\]
\end{enumerate}
The function $V$ is said to be a viscosity solution of the HJB equation~\eqref{eq.HJB} if it is both a viscosity supersolution and a viscosity subsolution.
\end{definition}

A function $\phi$ as in~\eqref{testfunction-1} or~\eqref{testfunction-2} in Definition~\ref{def:viscosity-solution} is said to be a test function for $V$ at $x$.

\begin{theorem}\label{thm:viscosity-solution}
The value function $V_c$ is a viscosity solution of the HJB equation~\eqref{eq.HJB} with boundary condition~\eqref{eq.boundary_Vc}.
\end{theorem}

The proof of the above
theorem is rather standard in the literature (see, e.g., \cite[Appendix B.1]{guo-pham_2005}), except for the boundary condition~\eqref{eq.boundary_Vc}, which we have already provided in Proposition~\ref{prop.Vcprime}.

\subsection{Concavity and a free-boundary value problem}\label{sect:concavity}

In our approach, proving the concavity of the value function $V_c$ is instrumental. In particular, it ensures differentiability and it will suggest an optimal strategy through its derivative. However, to prove concavity, we cannot apply the standard argument consisting in taking a convex combination of arbitrary strategies (see, e.g., Lemma~4.3 in \cite{guo-pham_2005}) because we consider a general diffusion process. To prove that $V_c$ is concave, we will show that its concave envelope is a subsolution of the HJB equation~\eqref{eq.HJB} and then invoke our comparison result in Lemma~\ref{lemma.sub_comp}.

\begin{proposition}
The value function $V_c$ is concave on $[0,\infty)$.
\end{proposition}

\begin{proof}
Let us first recall that $V_c$ is continuous on $[0,+\infty)$ by Proposition~\ref{prop:immediate-injection}.
Let $U$ be the concave envelope of $V_c$ on $[0,\infty)$. Let us first prove by contradiction that $U(0)=V_c(0)$ and $U^\prime(0) \geq V_c^\prime(0)$.

Suppose that $U(0) > V_c(0)$. By definition of the concave envelope and continuity, there exists $x_0 > 0$ such that $U$ is affine on $[0,x_0]$ and $U(z) > V_c(z)$ for all $z \in [0,x_0]$. Let $d := \min_{y \in [0,x_0]}\PAR{ U(y) - V_c(y) } > 0$,
and define $L$ as the affine function on $[0,x_0]$ such that $L(0) = U(0)-d$ and $L(x_0) = U(x_0)$.\\
By definition, we have $L \leq U$ and $L > U - d$ on $(0,x_0]$, and that the slope of $L$ is larger than the slope of $U$. In addition, for $z \in (0,x_0]$  we have 
\begin{align*}
L(z) &> U(z) - d> U(z) - (U(z) - V_c(z))= V_c(z),
\end{align*}
where in the inequality we have simply used the definition of $d$ as the minimal distance between $U$ and $V_c$. In addition, for $z=0$, $L(0)=U(0)-d\geq V_c(0)$ by definition of $d$. Thus, if we consider the function
\[
\hat{U}(z) = \begin{cases}
    L(z), & z \in [0,x_0]\\
    U(z), & z \in (x_0,\infty), 
\end{cases}
\]
we obtain that $\hat{U}$ is concave, $\hat{U}\leq U$ with $\hat U<U$ on $[0,x_0)$, and $\hat{U}\geq V_c$, which is a contradiction that $U$ is the concave envelope. Hence, $U(0) = V_c(0)$.\\
Knowing that $U(0) = V_c(0)$, we prove that $U'(0) \geq V_c'(0)$. Indeed, suppose that $U'(0) < V_c'(0)$, then by the continuity of $U$ and $V_c$ and the fact that $U(0) = V_c(0)$, there exists $\delta > 0$ such that $U \leq V_c$ on $[0,\delta]$, thus contradicting that $U$ is the concave envelope.

\medskip

By definition, $U \geq V_c$ on $[0,\infty)$, so it suffices to show that $U \leq V_c$ on $[0,\infty)$. By contradiction, assume there exists $x_0 \in (0,\infty)$ such that $U(x_0) > V_c(x_0)$. Therefore, by continuity of both functions, there exists $0 \leq x_1 < x_0 <  x_2 \leq \infty$ such that: $U(x_1)=V_c(x_1)$ and $U(z)>V_c(z)$ for all $z \in (x_1,x_2)$. If $x_2<\infty$, then we further have $U(x_2)=V_c(x_2)$. In both cases, by Lemma B.1 in \cite{guerin-et-al_2024}, the concave envelope is affine on $(x_1,x_2)$. Let us split the rest of the proof in two parts, whether $x_2<\infty$ or $x_2=\infty$.

\underline{First case}: let us assume that $x_2<\infty$, then we have that $U(z) = V_c(x_1)+p(z-x_1)$ for $z \in [x_1,x_2]$, and with $p=\frac{V_c(x_2)-V_c(x_1)}{x_2-x_1}$. By Proposition~\ref{prop:immediate-injection}, $\beta(x_2-x_0) \geq V_c(x_2) - V_c(x_0).$ Therefore, $\beta(x_2-x_0) > U(x_2) - U(x_0) = p (x_2-x_0),$ which implies that $p < \beta$.\\
Let us show  by contradiction that $x_1>0$. Indeed, if $x_1=0$ then $U(z)=V_c(0)+pz$ for all $z \in (0,x_2)$. In particular, we would have $p=U^\prime(0) \geq V_c^\prime(0) =\beta$, where the last equality comes from Proposition~\ref{prop.Vcprime}. This is a contradiction with the above result saying that $p<\beta$.

Having established that $x_1 >0$, we prove that $U$ is an affine subsolution on $[x_1,x_2]$. First, we prove that the subsolution property is satisfied at $x_1$ and $x_2$. Let $\epsilon \in (0,x_1)$ be small enough, and $L:[0,\infty) \to \R$ be the affine function with slope $p$ and such that $L(x_1) = U(x_1)$ (hence $L(x_1) = V_c(x_1)$). Since $U$ is concave and $p=U'(x_1+)$, we have that $L \geq U \geq V_c$, hence $L$ is a test function for both $V_c$ and $U$ at $x_1$. Since $V_c$ is a viscosity subsolution, we have that
\begin{equation}\label{ineq.subsolution.L}
q L(x_1) - \sup_{0 \leq \ell \leq F(x_1)} \mathcal L_\ell [L](x_1) \leq 0.
\end{equation}
We follow a similar argument to prove inequality \eqref{ineq.subsolution.L} at $x_2$.

\medskip

Finally, we prove that inequality \eqref{ineq.subsolution.L} also holds on $(x_1,x_2)$. For a fixed $z \in (x_1,x_2)$, there exists $\lambda \in (0,1)$ such that $z=\lambda x_1 + (1-\lambda) x_2$. Then, we can write, using that both $\mu$ and $F$ are concave,
\begin{align*}
q U(z) - \sup_{0 \leq \ell \leq F(z)} \mathcal L_\ell [U](z) &= q U(z) - p \mu(z) - (1-p) \Ind_{p<1} F(z) \\
&\leq q \left( \lambda U(x_1) + (1-\lambda) U(x_2) \right) - p \left( \lambda \mu(x_1) + (1-\lambda) \mu(x_2) \right) \\
& \qquad \qquad \qquad - (1-p) \Ind_{p<1} \left( \lambda F(x_1) + (1-\lambda) F(x_2) \right) \\
&= \lambda \left( q U(x_1) - \sup_{0 \leq \ell \leq F(x_1)} \mathcal L_\ell [U](x_1) \right) \\
& \qquad \qquad + (1-\lambda) \left( q U(x_2) - \sup_{0 \leq \ell \leq F(x_2)} \mathcal L_\ell [U](x_2) \right) \\
&\leq 0 .
\end{align*}
In other words, $U$ is an affine subsolution on $(x_1,x_2)$. The contradiction then follows from Lemma~\ref{lemma.sub_comp}.

\underline{Second case}: let us assume that $x_2=\infty$. We then have that $U(z)=V_c(x_1)+p(z-x_1)$ for all $z \in [x_1,\infty)$, with $p=B(\infty) \in [0,1)$, where $B$ is defined in~\eqref{def-of-B} (see Remark~\ref{rk:B(infinity)}). Since $p<1<\beta$, we can deduce that $x_1>0$. Let us note that Inequality \eqref{ineq.subsolution.L} still holds. Now, for a fixed $z \in (x_1,\infty)$, we can write, using again that both $\mu$ and $F$ are concave,

\begin{align*}
q U(z) - \sup_{0 \leq \ell \leq F(z)} \mathcal L_\ell [U](z) &= q U(z) - p \mu(z) - (1-p)  F(z) \\
&= q V_c(x_1) - p \mu(x_1) - (1-p) F(x_1) \\
& \qquad + p (z-x_1) \left[q - \frac{\mu(z)-\mu(x_1)}{z-x_1} - \left( \frac{1-p}{p} \right) \left( \frac{F(z)-F(x_1)}{z-x_1} \right) \right] \\
&\leq q U(x_1) - \sup_{0 \leq \ell \leq F(x_1)} \mathcal L_\ell [U](x_1) \\
& \qquad + (z-x_1) \left(p \left(q - \mu^\prime (z) \right) - (1-p) F^\prime (z) \right) \\
&\leq (z-x_1) \left( p \left(q - \mu^\prime (z) + F^\prime (z) \right) - F^\prime (z) \right) \\
&\leq 0 ,
\end{align*}
where we have used \eqref{ineq.subsolution.L} and  in the last step we used the fact that $p = B(\infty) \leq B(z)$, since $B$ is decreasing. By Lemma~\ref{lemma.sub_comp}, $U \leq V_c$ on $[x_1, \infty)$, which gives the contradiction. 
\end{proof}

Now, given that $V_c$ is concave on $[0,\infty)$, both its left derivative $V_c^{'-}$ and its right derivative $V_c^{'+}$ exist and are nonincreasing on $(0,\infty)$.

Now, let us define $a_c := \inf \{x \geq 0 \colon V_c^{'+}(x) < \beta\}$ and $b_c := \inf \{x\geq 0 \colon V_c^{'-}(x) < 1\}$. We have that $0 \leq a_c < b_c < \infty$. Indeed, by Proposition~\ref{prop.Vcprime}, we deduce that $V_c^{'+}(0)=\beta$, which gives the inequality $a_c \geq 0$. Also, by Proposition~\ref{prop.affine_bound}, we deduce the finiteness of $b_c$. In other words, we have $[a_c, b_c) \subset [0,\infty)$.

In addition, we can deduce that $V_c(x) = V_c(0) + \beta x = V_c(a_c) + \beta(x-a_c)$ for all $x \in [0,a_c]$. Indeed, for each $x \in (0,a_c)$, we have $\beta = V_c^{'+}(0) \geq V_c^{'-}(x) \geq V_c^{'+}(x) \geq \beta$. In other words, $V$ is continuously differentiable on $(0,a_c)$ and we have $V^\prime(x)=\beta$ for all $x \in [0,a_c)$.

\begin{proposition}\label{prop:vc:c1}
The value function $V_c$ is continuously differentiable on $(0,\infty)$. Also, it is such that $V_c(x) = V_c(0) + \beta x = V_c(a_c) + \beta(x-a_c)$ for all $x \in [0,a_c]$ and $1 < V_c^\prime (x) < \beta$ for all $x \in (a_c,b_c)$.
\end{proposition}

\begin{proof}
Since $V_c$ is concave, we have that $V_{c}^{\prime +}(x)$ and $V_{c}^{\prime -}(x)$ are well defined and such that $ V_{c}^{\prime +}(x) \leq V_{c}^{\prime -}(x) $ for all $x \in (0,\infty)$.

Suppose there exists $x \in (0,\infty)$ such that $V_{c}^{\prime +}(x) < V_{c}^{\prime -}(x)$. Take $d \in ( V_{c}^{\prime +}(x), V_{c}^{\prime -}(x) )$ and $M > 0$, and consider the function
\[
\phi(z) := V_c(x) + d (z-x) -\frac{M}{2}( z-x )^2, \quad z \geq 0.
\]
Note that $\phi(x) = V_c(x)$ and, since $d \in ( V_{c+}'(x), V_{c-}'(x) )$, by continuity there exists $\delta >0$ such that $\phi(z) \geq V_c(z)$ for all $z\in (x-\delta,x+\delta)$. It follows that $\phi$ is a test function for $V_c$ at $x$. 
Using that $V_c$ is a viscosity subsolution of \eqref{eq.HJB} and $\phi$ is a test function, we obtain
\[
qV_c(x) + \frac{\sigma^2(x)}{2}M - d \mu(x) - \sup_{0 \leq \ell \leq F(x)}\BRA{(1-d)\ell} \leq 0.
\]
However, the previous inequality is not satisfied if we take $M$ large enough, yielding a contradiction. Hence, $V_c \in \cC^1(0,\infty)$.
Finally, from the fact that $V_c'$ is continuous, we further have that
\[
a_c = \inf \{x \geq 0 \colon V_c^\prime (x) < \beta\} \quad \text{and} \quad b_c = \inf \{x \geq 0 \colon V_c^\prime(x) < 1\} .
\]
It follows that $V_c'(x) \in (1,\beta)$ for all $x \in (a_c, b_c)$.
\end{proof}

From the proof of this last proposition, we now have that
\[
a_c = \inf \{x \geq 0 \colon V_c^\prime (x) < \beta\} \quad \text{and} \quad b_c = \inf \{x \geq 0 \colon V_c^\prime(x) < 1\} .
\]

The next step is to show that $V_c$ is the unique classical solution to a \textit{free-boundary value problem}. In this direction, let us define 
\[
\mathcal W := \BRA{v \colon [0,\infty) \to \mathbb{R} : v \text{ Lipschitz continuous} } .
\]
Note that for $v\in\mathcal{W}$, we have $ \sup_{x\geq 0} \frac{|v(x)|}{1+x}<\infty$. Note also that $V_c\in\mathcal{W}$ by Proposition~\ref{prop:immediate-injection}.

Here is a first version of the free-boundary value problem of interest:
\begin{align}
&(\mathcal{L} - q) [v](x) = 0, & \text{for } x \in (a_c, b_c), \label{HJB:ODE:1}\\
&(\mathcal{L} - q ) [v](x) + (1-v'(x))F(x) = 0, &\text{for } x \in (b_c, \infty), \label{HJB:ODE:2}\\
&v'(b_c) = 1, & \label{HJB:ODE:3} \\
&v'(a_c) = \beta . &  \label{HJB:ODE:4}
\end{align}

\begin{theorem}\label{thm:classical:solution}
The restriction of $V_c$ to $[a_c,b_c]$ is the unique (in $\mathcal{W}$) classical solution of~\eqref{HJB:ODE:1} with boundary conditions~\eqref{HJB:ODE:3}-\eqref{HJB:ODE:4}, while the restriction of $V_c$ to $[b_c,\infty)$ is the unique (in $\mathcal{W}$) classical solution of~\eqref{HJB:ODE:2} with boundary condition~\eqref{HJB:ODE:3}.

Also, we have $V_c''(b_c-) = V_c''(b_c+)$ and $V_c''(a_c-)=0$, and $V_c\in \mathcal{C}^2(0,\infty)$.
\end{theorem}
\begin{proof}
Since $1< V_c' < \beta$ on $(a_c,b_c)$, we readily know that $\sup_{0 \leq \ell \leq F(x)} (1-V_c'(x)) \ell = 0$ for all $x\in (a_c,b_c)$. Therefore, by Theorem~\ref{thm:viscosity-solution}, $V_c$ is a viscosity solution of~\eqref{HJB:ODE:1} on $(a_c,b_c)$. By Theorem~3.3 of \cite{Crandall1992}, the viscosity solution of \eqref{HJB:ODE:1} with boundary conditions $v(a_c) = V_c(a_c)$ and $v(b_c) = V_c(b_c)$ is unique. By standard ODE results, there also exists a unique $\mathcal{C}^2$ solution to this equation with these boundary conditions. Since classical solutions are also viscosity solutions, we can conclude that $V_c$ is in fact $\mathcal{C}^2$ on $(a_c,b_c).$ This proves the first statement since $V_c$ satisfies \eqref{HJB:ODE:3}-\eqref{HJB:ODE:4}.

Since $ V_c' < 1$ on $(b_c, \infty)$, we readily know that $\sup_{\ell \leq F(x)} (1-V_c'(x)) \ell = (1-V_c'(x)) F(x)$ for all $x\in (b_c,\infty).$ Therefore, $V_c$ is a viscosity solution of \eqref{HJB:ODE:2} on $(b_c,\infty)$.

Let $\bar b > b_c$. Using the same uniqueness argument with ODE \eqref{HJB:ODE:2} on the interval $(b_c,\bar b)$ with boundary conditions $v(b_c) = V_c(b_c)$ and $v(\bar b) = V_c(\bar b),$ we can also conclude that $V_c$ is $\mathcal{C}^2$ on $(b_c,\infty),$ since $\bar b$ is arbitrary. This  proves the second statement since $V_c \in \mathcal W.$

We now verify that $V_c$ is twice differentiable at $b_c$. Recall that for $x \in (a_c, b_c)$ we have that $V_c$ satisfies the ODE \eqref{HJB:ODE:1}, then by taking the limit as $x \uparrow b_c$, and using that $V_c'(b_c) = 1$ we get
    \[
    \frac{\sigma^2(b_c)}{2}V_c''(b_c-) + \mu(b_c) - qV_c(b_c) = 0,
    \]
    equivalently
    \[
    V_c''(b_c-) = \frac{2}{\sigma^2(b_c)}\PAR{qV_c(b_c) - \mu(b_c)}.
    \]
    Similarly, for $x \in (b_c,\infty)$ we have that $V_c$ satisfies \eqref{HJB:ODE:2}, then by taking the limit as $x \downarrow b_c$, and using that $V_c'(b_c) = 1$ we obtain
    \[
    \frac{\sigma^2(b_c)}{2}V_c''(b_c+) + \mu(b_c) - qV_c(b_c) = 0,
    \]
    which implies
    \[
    V_c''(b_c+) = \frac{2}{\sigma^2(b_c)}\PAR{qV_c(b_c) - \mu(b_c)}.
    \]
    It now follows easily that $V_c''(b_c-) = V_c''(b_c+)$.\\
    Now, since $V_c$ is affine on $[0,a_c)$ we have that $V_c''(a_c-) = 0$. In addition, we have 
    \[
    q V_c(a_c) - \beta \mu(a_c) \geq 0,
    \] since $V_c$ is a supersolution of \eqref{eq.HJB} and $V_c'(x) = \beta$ on $(0,a_c)$.
    On the other hand, by taking the limit in \eqref{HJB:ODE:1} as $x \downarrow a_c$, and using that $V_c'(a_c) = \beta$ we obtain
    \[
    \frac{\sigma^2(a_c)}{2}V_c''(a_c+) + \beta\mu(a_c) - qV_c(a_c) = 0,
    \]
    which implies
    \begin{align*}
    V_c''(a_c+) &= \frac{2}{\sigma^2(a_c)}\PAR{qV_c(a_c) - \beta\mu(a_c)} \geq 0.
    \end{align*}
    However, the fact that $V_c$ is concave implies that $V_c''(a_c+) = 0$. Thus, $V_c \in \mathcal{C}^2(0,\infty)$.
\end{proof}

Using the structure of the HJB equation, we can in fact prove that $a_c=0$ and, finally, specify the free-boundary value problem of which $V_c$ is the solution. 

\begin{corollary}\label{cor:a-star-zero}
The value function $V_c$ is the unique classical solution in $\mathcal{W}$ to the following free-boundary value problem:
\begin{align}
& (\mathcal{L} - q) [v](x) = 0, & \text{for $x \in (0, b_c)$,} \label{HJB:Simple:1}\\
        & ( \mathcal{L} - q ) [v](x) + (1-v'(x))F(x) = 0, & \text{for $x \in (b_c, \infty)$,} \label{HJB:Simple:2}\\
        & v'(b_c) = 1 , & \label{HJB:Simple:3} \\
        & v'(0) = \beta , & \label{HJB:Simple:4}
\end{align}
where $b_c=\inf\BRA{x\geq 0 \colon V'_c(x)=1}$.
\end{corollary}
\begin{proof} 
In view of Theorem~\ref{thm:classical:solution}, all is left to prove is that $a_c = \inf \{x \geq 0 \colon V_c^\prime(x) < \beta\} = 0$. Let us proceed by contradiction and assume $a_c>0$.

By Proposition~\ref{prop:vc:c1}, we can write $V_c(x) = V_c(a_c) + \beta (x-a_c)$, or equivalently $V_c(x) = V_c(0) + \beta x$, for $0 \leq x \leq a_c$. By Theorem~\ref{thm:classical:solution}, this implies that $V_c''(x) = 0$ for $0 \leq x \leq a_c$. As a consequence, since $V_c$ is a supersolution of the HJB~\eqref{eq.HJB} on $(0,\infty)$, for $0 \leq x \leq a_c$, we have
\[
0 \leq \min\{q V_c(x) - \sup_{0 \leq \ell \leq F(x)} \mathcal L_\ell [V_c](x) , \beta - V_c'(x)\} = q V_c(x) - \mu(x) \beta.
\]
Furthermore, $q V_c(a_c) - \mu(a_c) \beta = 0$ because of~\eqref{HJB:ODE:1}. Then, using  the concavity of $\mu$ and Assumption~\ref{assumption:SDE}, we can also write, for $0 \leq x \leq a_c$,
\begin{align*}
q V_c(x) - \mu(x) \beta &= q \left( V_c(x) - V_c(a_c) \right) - \beta \left(\mu(x) - \mu(a_c) \right) \\
&=  -q \beta (a_c-x) - \beta \left(\mu(x) - \mu(a_c) \right) \\
&\leq -\beta \left( \mu'(a_c) (a_c-x) + \mu(x) - \mu(a_c) \right) \\
&\leq 0 .
\end{align*} 

Putting the pieces together, we have that $q V_c(x) - \mu(x) \beta = 0$ for all $0 \leq x \leq a_c$, from which we deduce that
\[
V_c'(x) = \frac{\mu'(x)}{q} \beta < \beta.
\]
This is a contradiction.
\end{proof}

From the above, we can now say that the value function $V_c$ is a classical solution to the HJB equation~\eqref{eq.HJB} with boundary
condition $V_c^\prime(0)=\beta$. Since $V_c$ is concave, and therefore has a nonincreasing derivative, we can further simplify the HJB equation~\eqref{eq.HJB}, as 
\begin{equation}\label{eq.HJB-simple}
q v(x) - \sup_{0 \leq \ell \leq F(x)} \mathcal L_\ell [v](x) = 0 , \quad x \in (0,\infty) ,
\end{equation}
for $v \in \mathcal{W}$ subject to the boundary condition $v^\prime(0)=\beta$.

\section{An optimal refraction-reflection strategy in $\Pi_c$}\label{Sec:Optimal:1}

The free-boundary value problem in Corollary~\ref{cor:a-star-zero}, of which $V_c$ is the solution, suggests that an optimal strategy consists in reflecting the process at zero (for capital injections) and refracting the process at the maximal rate at level $b_c$ (for dividend payments). In this section, we verify this by providing a probabilistic representation of this optimal policy.

As a first step, we will define and then find an expression for the performance function of an arbitrary \textit{refracted-reflected strategy}, also called a \textit{reflected mean-reverting strategy} in \cite{RRS2026}. In general, such a strategy would be parametrized by a refraction threshold and a reflection barrier. However, in view of the above discussion, only one parameter will be needed, which is the one for refraction, as the reflection barrier will be set to zero. As a second step, we will show that the performance function of the refracted-reflected strategy with refraction threshold $b_c$ (defined in the previous section) is equal to the value function function $V_c$, therefore proving that this strategy is optimal in $\Pi_c$. As a by-product, we will obtain a semi-explicit characterization of the optimal threshold $b_c$.

\subsection{Refraction-reflection strategies}\label{sec:An arbitrary reflection-refraction strategy}

Fix $b>0$. We define the refraction-reflection strategy at level $b$ by $\pi^b=(\ell^b, K^b)$ through the dynamics
\begin{equation}
\diff X_t^b = \left(\mu(X_t^b) - F(X_t^b)\Ind_{\{X_t^b \geq b \}}\right) \diff t + \sigma(X_t^b) \diff W_t + \diff K^b_t ,
\end{equation}
therefore setting $\ell^b_t := F(X_t^b)\Ind_{\{X_t^b \geq b \}}$ and $K^b_t := \int_{(0,t)} \Ind_{\{X^b_s = 0 \}} \diff K^b_s$, for all $t\geq 0$. Note that $X^b_t \geq 0$ for all $t \geq 0$, which means that $\pi^b \in \Pi_c$ or, said differently, that a refraction-reflection strategy is an admissible strategy. Mathematically, the corresponding controlled process $X^b$ is a solution to a \textit{Skorohod reflection problem}; we refer the reader to Section~1.3 in \cite{pilipenko14} for an in-depth discussion on this.

Note that it suffices to consider candidate refraction-reflection strategies with a refraction threshold $b>0$ since $b_c=\inf\BRA{x\geq 0 \colon V'_c(x)=1}>0$.

\begin{remark}
 In the above, if we set $K^b \equiv 0$, then we obtain the pure refraction strategy at threshold $b$, namely the strategy $(\ell^b, 0) \in \Pi_d$. 
\end{remark}

For ease of notation, we define the performance function of $\pi^b$ by $J_b(x) := J(x; \pi^b)$, which means that we have
\begin{equation}\label{eq.Jb}
J_{b}(x) = \E_x\SBRA{\int_0^\infty e^{-qt} F(X^b_t) \Ind_{\BRA{X^b_t \geq b}} \diff t - \beta \int_{[0,\infty)} e^{-qt} \diff K^b_t} .
\end{equation}

Now, let us compute $J_{b}$ using probabilistic decompositions and fluctuation identities for diffusion processes. First, for a fixed $x \in [0,b]$, we have that $\BRA{(X^b_t, K^b_t), 0 \leq t \leq \tau_{b}^{X^b}}$, with $X^b_{0-}=x$, is equal in law to $\BRA{( Y_t, R_t), 0 \leq t \leq \tau^Y_b}$, where $(Y_t,R_t)$ is the solution to
\begin{equation}
\diff Y_t = \mu(Y_t) \diff t + \sigma(Y_t) \diff W_t + \diff R_t ,
\end{equation}
with $Y_{0-}=x$, in which $R_t = \int_{(0,t)} \Ind_{\{Y_s = 0 \}} \diff R_s$ for all $t\geq 0$, and where
\[
\tau^{X^b}_b := \inf \BRA{t > 0 \colon X^b_t = b} \quad \text{and} \quad \tau^Y_b := \inf \BRA{t > 0 \colon Y_t = b} .
\]
Using the strong Markov property of the process $X^b$, we can write
\[
J_{b}(x) = J_{b}(b) \E_x\SBRA{e^{-q \tau^Y_b}} - \beta \E_x \SBRA{\int_{[0,\tau^Y_b)} e^{-qt} \diff R_t} .
\]
Second, for a fixed $x \in (b,\infty)$, we have that $\BRA{(X^b_t, K^b_t), 0 \leq t \leq \tau_b^{X^b}}$, with $X^b_{0-}=x$, is equal in law to $\BRA{X_t,\, 0 \leq t \leq \tau^X_b}$, where
\begin{equation}
\diff X_t = \left(\mu(X_t) - F(X_t) \right) \diff t + \sigma(X_t) \diff W_t ,
\end{equation}
with $X_0=x$, where
\[
\tau^X_b := \inf \BRA{t > 0 \colon X_t = b} .
\]
Using the strong Markov property of the process $X^b$, we can write
\[
J_{b}(x) = J_{b}(b) \E_x\SBRA{e^{-q \tau^X_b}} + \E_x \SBRA{\int_0^{\tau^X_b} e^{-qt} F(X_t) \diff t} .
\]
Putting these two pieces together, we obtain
\begin{equation}
J_{b}(x) =
\begin{cases}
J_{b}(b) \E_x\SBRA{e^{-q \tau^Y_b}} - \beta \E_x \SBRA{\int_{[0,\tau^Y_b)} e^{-qt} \diff R_t} & \text{for $x \in [0,b]$,}\\
J_{b}(b) \E_x\SBRA{e^{-q \tau^X_b}} + \E_x \SBRA{\int_0^{\tau^X_b} e^{-qt} F(X_t) \diff t} & \text{for $x \in (b,\infty)$.}
\end{cases}
\end{equation}

Now, we need to compute the above expectations and the value of $J_b(b)$. Let $\psi$ (resp.\ $\varphi$) be the positive and increasing (resp.\ decreasing) fundamental solutions to
\[
(\Gen - q)[v](x) = 0 , \; x \in (0,\infty) ,
\]
with boundary conditions $\psi(0)=0$ and $\psi^{\prime}(0)=1$ (resp.\ $\varphi(0) = 1$ and $\lim_{x\to\infty} \varphi(x) = 0$), where $\cL$ was defined in~\eqref{eq.generatorL} with $\ell=0$.

It is known (see Corollary~2.2 in \cite{SLG1984}) that 
\begin{equation}\label{fluctuation:ODE}
g_b(x) := \E_x \SBRA{e^{-q \tau^Y_b}} \quad \text{and} \quad h_b(x) := \E_x \SBRA{\int_{[0,\tau^Y_b)} e^{-qt} \diff R_t}
\end{equation} 
are solutions of
\[
(\Gen - q)[v](x)=0 , \; x \in (0,b) ,
\]
with boundary conditions $g_b'(0)=0$, $g_b(b) =1 $, and $h_b'(0) = -1$, $h_b(b) =0$.
\begin{remark}
We remark that the functions $g_b$ and $h_b$ depend strongly on $b$. To make this evident, note that we can express them in terms of the fundamental solutions as follows: for $x \in [0,b]$,
    \begin{align}
g_b(x) = \frac{\varphi'(0) \psi(x) - \psi'(0) \varphi(x)}{\psi(b) \varphi'(0) - \psi'(0) \varphi(b)} , \quad h_b(x) = \frac{\varphi(x)\psi(b) - \varphi(b)\psi(x)}{\psi'(0) \varphi(b) - \psi(b) \varphi'(0)} . \label{def:g-h}
\end{align}
\end{remark}
 
Let $\varphi_F$ be the decreasing fundamental solution of
\[
(\Gen - q)[v](x) - F(x) v'(x) = 0 , \; x \in (0,\infty) ,
\]
with boundary conditions $\varphi_F(0)=1$ and $\lim_{x\to\infty}\varphi_F(x)=0$. Hence, for $x \in (b,\infty)$, we have $\E_x\SBRA{e^{-q \tau^X_b}} = \frac{\varphi_F(x)}{\varphi_F(b)}$.

Also, extending the function $F$ on $(-\infty,0)$ as in \cite{guerin-et-al_2024},  it is known that 
\[
I_F(x) := \E_x \SBRA{\int_0^{\infty} e^{-qt} F(X_t) \diff t}
\] 
is a solution in $\mathcal W$ of
\[
(\Gen_F - q)[v](x)  = 0 , \; x \in (0,\infty) .
\]
We refer the reader to Section~3 of \cite{guerin-et-al_2024} for more information on the analytical properties of $I_F$ under Assumption~\ref{assumption:SDE}. Note that no assumption on the sign of $\mu(0)$ is needed for the study of the properties of $I_F$. The assumption $\mu(0)>0$ is used in \cite{guerin-et-al_2024} to ensure the existence of a unique inflection point of $\psi$ (see also \cite[Lemma~2.4]{ekstrom-lindensjo_2023}), and therefore to prove the existence of an optimal refraction strategy when capital injections are not allowed.

Putting the pieces above together, we obtain
\begin{equation}
J_{b}(x) =
\begin{cases}
J_{b}(b) g_b(x) - \beta h_b(x) & \text{for $x \in [0,b]$,}\\
I_F(x) + \frac{\varphi_{F}(x)}{\varphi_{F}(b)}(J_{b}(b) - I_F(b)) & \text{for $x \in (b,\infty)$.}
\end{cases}
\end{equation}
It can be seen easily that $J_{b}$ is continuous on $[0,\infty)$. In addition, the constant $J_{b}(b)$ is chosen so that $J'_{b}(b-) = J'_{b}(b+)$, and its value can be computed explicitly in terms of the functions $g_b,h_b,I_F$ and $\varphi_{F}$.

Summing up, we have
\begin{equation}\label{J:Fluctuation}
J_{b}(x) = \begin{cases}
A g_b(x) - \beta h_b(x) & \text{for $x \in [0,b]$,}\\
I_F(x) + B \varphi_{F}(x) & \text{for $x \in (b,\infty)$,}
\end{cases}
\end{equation}
where the constants $A,B$ are such that $J_{b} \in \mathcal{C}^1(0,\infty)$.

\begin{remark}
Note that, in general, $J_{b}$ is not twice continuously differentiable on $(0,\infty)$. However, it is twice continuously differentiable on $(0,b)$ and on $(b,\infty)$, and it is a solution to
\[
(\cL - q)[v](x) = 0, \quad x \in (0,b),
\]
and
\[
(\cL_F - q)[v](x) = 0, \quad x \in (b,\infty).
\]
\end{remark}

\begin{remark}\label{rk:V_c finite}
It is now clear from the above computations that $V_c(x)>-\infty$ for any $x \geq 0$, as stated just before Proposition~\ref{prop:immediate-injection}. Indeed, for any $b$, we have $\pi^b \in \Pi_c$ and therefore $V_c(x) \geq J_b(x)$.
\end{remark}

\subsection{Optimality of the refraction-reflection strategy at level $b_c$}\label{sec:optimality-reflection-refraction}

Now, we can establish the connection between the value function $V_c$ and the performance function $J_{b_c}$ of $\pi^{b_c}$, the refraction-reflection strategy at level $b_c=\inf\BRA{x\geq 0:V'_c(x)=1}$.

\begin{proposition}\label{prop:explicit-Vc}
The refraction-reflection strategy at level $b_c$ is optimal in $\Pi_c$, i.e., $V_c(x)=J_{b_c}(x)$ for all $x \in [0,\infty)$. As a consequence, the value function $V_c$ is given by
    \[
    V_c(x)=\begin{cases}
    \frac{1+\beta \bar{h}'(b_c)}{\bar{g}'(b_c)}\bar{g}(x)-\beta \bar{h}(x)&\text{for }x\in[0,b_c],\\
    I_F(x)+\frac{1-I'_F(b_c)}{\varphi_F'(b_c)}\varphi_F(x)&\text{for }x\in[b_c,\infty),
    \end{cases}
    \]
     where $\bar{g}(x):=g_{b_c}(x)$ and $\bar{h}(x):=h_{b_c}(x)$.
\end{proposition}
\begin{proof}
First, note that the restriction of $V_c$ to $[0, b_c]$ is the unique solution of $(\cL - q)[v](x) = 0$ with $v'(0) = \beta$ and $v'(b_c)=1$. It follows that
\[
V_c(x) = C \bar{g}(x) -\beta \bar{h}(x), \quad x \in [0,b_c),
\]
for a suitable constant $C$. 

Second, we know that the restriction of $V_c$ to $(b_c,\infty)$ is the unique solution in $\mathcal{W}$ to $(\cL_F - q)[v](x) = 0$ with $v'(b_c)=1$. It follows that
\[
V_c(x) = I_F(x) + D \varphi_{F}(x), \quad x \in (b_c,\infty),
\]
for a suitable constant $D$. Since $V_c \in \mathcal{C}^1(0,\infty)$, then $C,D$ satisfy the system
\begin{align}
C \bar{g}(b_c) -\beta \bar{h}(b_c) &=I_F(b_c) + D \varphi_{F}(b_c) , \label{C1:1}\\
C \bar{g}'(b_c) -\beta \bar{h}'(b_c) &=I_F'(b_c) + D \varphi_{F}'(b_c) .\label{C1:2}
\end{align}

On the other hand, we have that $J_{b_c}$ is given by~\eqref{J:Fluctuation} and it is continuously differentiable, so that $A,B$ in~\eqref{J:Fluctuation} also solve \eqref{C1:1}--\eqref{C1:2}. We conclude that $(A,B) = (C,D)$, hence $V_c = J_{b_c}$. Finally, since we also know that $V_c'(b_c) = 1$, we can further deduce that \begin{align*}
A  = \frac{1 + \beta \bar{h}'(b_c)}{\bar{g}'(b_c)}, & \mbox{ and } B  = \frac{1 - I_F'(b_c)}{\varphi_{F}'(b_c)}.
\end{align*}

\end{proof}

\section{Optimization of capital injections and dividend payments}\label{sec:Optimal:2}

We now turn our attention to our general optimization problem, in which we are optimizing both dividends and capital injections over the set of admissible strategies $\Pi$, using the performance function defined in~\eqref{eq:performance-function}. Recall that its value function is given by $V$. In particular, as mentioned earlier, since $\Pi_d \cup \Pi_c \subsetneq \Pi$, we have
\[
V(x) \geq \max\BRA{V_d(x),V_c(x)}, \quad x \geq 0.
\]

\subsection{Maximization of dividends without injections}

As presented in Section~\ref{sect:model-problem}, the optimization problem with value function $V_d$, corresponding to optimization over strategies in $\Pi_d$, that is, strategies for which no capital injections are made, was studied in \cite{guerin-et-al_2024} under the additional assumption that $\mu(0)>0$. This assumption was used to establish the existence of an optimal refraction strategy in $\Pi_d$. Moreover, it was shown that the value function $V_d$ is concave on $[0,\infty)$ and is a classical solution to the following HJB equation:
\begin{equation}\label{eq.HJB_d}
q v(x) - \sup_{0 \leq \ell \leq F(x)} \mathcal L_\ell [v](x) = 0 , \quad x \in (0,\infty) ,
\end{equation}
with boundary condition $v(0)=0$. 

\begin{remark}
Recall from the discussion around HJB equation~\eqref{eq.HJB-simple}, which is the same as the one above in~\eqref{eq.HJB_d}, that the value function $V_c$ is also a classical solution to this HJB equation, but with boundary condition $v^\prime(0)=\beta$.
\end{remark}

In fact, the main result in \cite{guerin-et-al_2024} is the following: 
\begin{theorem}[Theorem~2.5 in~\cite{guerin-et-al_2024}]\label{thm.Vd}
Under Assumption~\ref{assumption:SDE} and if $\mu(0)>0$, then there exists a unique level $b_d \ge  0$ such that the pure refraction strategy at level $b_d$ is optimal in $\Pi_d$.
Also, $V_d$ is a twice continuously differentiable function such that:
\begin{itemize}
\item If $I_F^\prime(0) - I_F(0) \varphi_F^\prime(0) > 1$, then $b_d>0$ and the value function $V_d$ is a concave function given by
\[
V_d(x) =
\begin{cases}
\frac{\psi(x)}{\psi^\prime(b_d)} & \text{for $x \le b_d$,}\\
I_F(x) - \varphi_F(x) \left( \frac{1-I_F^\prime(b_d)}{\varphi_F^\prime(b_d)} \right) & \text{for $x \ge b_d$}.
\end{cases}
\]
\item If $I_F^\prime(0) - I_F(0) \varphi_F^\prime(0) \le  1$, then $b_d = 0$ and 
\[
V_d(x) = I_F(x) - I_F(0) \varphi_F(x), \quad x \ge 0 .
\]
\end{itemize}
\end{theorem}

Similarly to Corollary~\ref{cor:a-star-zero} above (concerning $V_c$), when $I_F^\prime(0) - I_F(0) \varphi_F^\prime(0) > 1$, we can deduce from the previous theorem that the value function $V_d$ is a classical solution to the following free-boundary value problem:
\begin{align}
& (\mathcal{L} - q) [v](x) = 0, & \text{for $x \in (0, b_d)$,}\\
        & ( \mathcal{L} - q ) [v](x) + (1-v'(x))F(x) = 0, & \text{for $x \in (b_d, \infty)$,}\\
        & v'(b_d) = 1 , & \\
        & v(0) = 0 , &
\end{align}
where $b_d=\inf\BRA{x\geq 0 \colon V'_d(x)=1}>0$. 

\begin{remark}
It is known also that, in the case $I_F^\prime(0) - I_F(0) \varphi_F^\prime(0) > 1$, the optimal level $b_d>0$ is the root of
\[
\frac{\psi(b)}{\psi^\prime(b)} = I_F(b) + \left(1- I_F^\prime(b) \right) \frac{\varphi_F(b)}{\varphi_F^\prime(b)} .
\]
See \cite{guerin-et-al_2024} for more details.
\end{remark}

\subsection{Solution to the general optimization problem}

The main result of this paper is the following dichotomous solution to the general optimization problem:
\begin{theorem}\label{thm.dichotomy}
Under Assumption~\ref{assumption:SDE} and if $\mu(0)>0$, then the value function $V$ is characterized by the following dichotomy:
    \begin{enumerate}
        \item If $V_c(0) > 0$, then $V=V_c$ and consequently an optimal dividend-injection strategy is given by the refraction-reflection strategy at threshold $b_c$.
        \item If $V_c(0) < 0$, then $V=V_d$ and consequently an optimal dividend-injection strategy is given by the pure refraction strategy at threshold $b_d$.
    \end{enumerate}
Equivalently,
\begin{enumerate}[resume]
        \item If $V_d'(0) > \beta$, then $V=V_c$ and consequently an optimal dividend-injection strategy is given by the refraction-reflection strategy at threshold $b_c$.
        \item If $V_d'(0) < \beta$, then $V=V_d$ and consequently an optimal dividend-injection strategy is given by the pure refraction strategy at threshold $b_d$.
    \end{enumerate}
In particular, if $V_d'(0) = \beta$ or $V_c(0)=0$, then we have $V=V_d=V_c$.
\end{theorem}

\subsection{Proof of Theorem~\ref{thm.dichotomy}}

We begin with a result similar to Lemma~\ref{lemma.super_comp}, i.e., a comparison result between classical supersolutions of HJB equation~\eqref{eq.HJB} and $V$. In fact, the proof is the same if we replace $\Pi_c$ by $\Pi$ and $V_c$ by $V$.

\begin{lemma}\label{lemma.super_compV}
If $\psi$ is a nonnegative classical supersolution of the HJB equation~\eqref{eq.HJB} and if $\psi(0) \geq V(0)$, then $\psi(x) \geq V(x)$ for all $x \in [0,\infty)$. 
\end{lemma}

Before proving the dichotomy stated in Theorem~\ref{thm.dichotomy}, we establish in the next two propositions a dichotomy based on the value of $V(0)$. Note that this is only an intermediate step, since the value of $V(0)$ is unknown at this stage.

\begin{proposition}\label{prop:Vc00}
If $V(0)=0$, then $V=V_d$.
\end{proposition}
\begin{proof}
It suffices to show that $V_d \geq V$. If $V(0) = 0 = V_d(0)$, then $V_d(x) - V_d(0) \leq V(x) - V(0) \leq \beta x$ for all $x \geq 0$, where the last inequality comes from Proposition~\ref{prop:immediate-injection}. Consequently, we have $V'_d(0) \leq \beta$ and, since by Theorem~\ref{thm.Vd} we have that $V_d$ is concave, then $V_d'(x) \leq \beta$ for all $x\geq 0$. In particular, this means that $V_d$ is also a nonnegative classical supersolution of HJB  equation~\eqref{eq.HJB}. Then, by Lemma~\ref{lemma.super_compV}, we have $V_d \geq V$.
\end{proof}

\begin{proposition}\label{prop:Vc0}
If $V(0)>0$, then $V=V_c$.
\end{proposition}

\begin{proof}
It suffices to show that $V_c\ge V$. First, we prove that $V_c(0)>0$ whenever $V(0)>0$.

Recall that, by Proposition~\ref{prop:explicit-Vc}, the supremum of the performance function over $\Pi_c$ is attained. More precisely, there exists $b_c>0$ such that, for the refraction-reflection strategy $\pi^{b_c}=(\ell^{b_c},K^{b_c})$, we have $V_c(0)=J(0;\pi^{b_c})$.

Let $\pi=(\ell,K)\in\Pi$ be arbitrary and set $\tau:=\tau_0^\pi$. On the canonical Skorohod space $\mathbb{D}(\R_+,\R)$, define the shifted path, on $\{\tau<\infty\}$, by
\[
(\Theta_\tau\omega)(t):=\omega(\tau(\omega)+t)-\omega(\tau(\omega)), \quad t\ge0.
\]
We define the strategy $\widehat\pi
=(\widehat\ell,\widehat K)$ by 
\[
\widehat \ell_t(\omega)
=
\begin{cases}
\ell_t(\omega), & t<\tau(\omega),\\
\ell^{b_c}_{t-\tau(\omega)}(\Theta_\tau\omega),
& t\ge \tau(\omega),
\end{cases}
\]
and
\[
\widehat K_t(\omega)
=
\begin{cases}
K_t(\omega), & t<\tau(\omega),\\
K_{\tau(\omega)}(\omega)
+
\Big(
K^{b_c}_{t-\tau(\omega)}(\Theta_\tau\omega)
-
K^{b_c}_0(\Theta_\tau\omega)
\Big),
& t\ge \tau(\omega).
\end{cases}
\]
We easily notice that $\widehat\pi$ belongs to $\Pi_c$.
Using the shift property  of the canonical process, we obtain
\begin{align*}
J(x;\widehat\pi)
 &=
\mathbb E_x\left[
\int_0^\tau e^{-qt}\ell_t\,\diff t
-\beta\int_{[0,\tau)} e^{-qt}\,\diff K_t
\right] + \e_x\SBRA{\e_x\SBRA{\int_{\tau}^{\infty} e^{-qt} \widehat\ell_t \diff t -  \beta \int_{[\tau,\infty)} e^{-qt} \diff \widehat K_t \Big| \mathcal{F}_{\tau}} 
}
 \\ &= \mathbb E_x\left[
\int_0^\tau e^{-qt}\ell_t\,\diff t
-\beta\int_{[0,\tau)} e^{-qt}\,\diff K_t
\right] + 
\mathbb E_x\left[
\ind_{\{\tau<\infty\}}e^{-q\tau}
J(0;\pi^{b_c})
\right] \\
&=
J(x;\pi)
+
\mathbb E_x\left[
\ind_{\{\tau<\infty\}}e^{-q\tau}
\right]V_c(0).
\end{align*}
Since $\widehat\pi\in\Pi_c$, it follows that $
V_c(x)\ge J(x;\widehat\pi)$,
and therefore
\begin{equation}\label{eq:comparison Vc and J}
V_c(x) \ge J(x;\pi) + \mathbb E_x \left[ \ind_{\{\tau<\infty\}}e^{-q\tau} \right] V_c(0) .
\end{equation}
Finally, if $x=0$, then we can write
\begin{equation}\label{eq:comparison Vc and J at zero}
\PAR{1 - \mathbb E_0 \left[ \ind_{\{\tau<\infty\}} e^{-q\tau} \right]} V_c(0) \ge J(0;\pi) .
\end{equation}
On the other hand, since we assume $V(0)>0$, there exists $\pi\in\Pi$ such that $J(0;\pi)>0$. Therefore, for this strategy, since we have $\p_0(\tau > 0)>0$, then
$\mathbb E_0 \left[ \ind_{\{\tau<\infty\}} e^{-q\tau} \right] < 1$. By the inequality in~\eqref{eq:comparison Vc and J at zero}, we deduce that $V_c(0)>0$.

Second, using~\eqref{eq:comparison Vc and J} with an arbitrary strategy $\pi\in\Pi$ and using the fact that $V_c(0)>0$, we have, for any $x \in [0,\infty)$, $V_c(x)\ge J(x;\pi)$. Taking the supremum over $\Pi$ yields the result.
\end{proof}

The final steps of the proof of Theorem~\ref{thm.dichotomy} are based on the following comparison principle for classical solutions.

\begin{theorem}[Comparison Principle]\label{th:comp}
Let $\phi, \psi \in \mathcal W$ be such that $\phi$ (resp.\ $\psi$) is a classical subsolution (resp.\ supersolution) of the HJB equation~\eqref{eq.HJB_d}. If $\phi(0) \leq \psi(0)$ or if $\phi'(0) \geq \psi'(0)$, then $\phi \leq \psi$ on $(0,\infty)$.
\end{theorem}

Such comparison results are standard in the theory of viscosity solutions (cf.\ \cite{Crandall1992}) and the classical theory of differential equations (cf.\ \cite{Protter1984}). They are generally obtained from maximum principles. For completeness, we provide a proof of Theorem~\ref{th:comp} in Appendix~\ref{sec:comparisonC2}. 

As the Comparison Principle in Theorem~\ref{th:comp} is for classical solutions, it can be applied to both value functions $V_d$ and $V_c$, since we now know they are both classical solutions of HJB equation~\eqref{eq.HJB_d}. In this direction, recall that
\[
V(0) \geq \max\BRA{V_d(0),V_c(0)}=\max\BRA{0,V_c(0)} .
\]

The rest of the proof of Theorem~\ref{thm.dichotomy} follows the numbering of the statement.
\begin{enumerate}
\item If $V_c(0) > 0$, then $V(0)\geq \max\{0, V_c(0)\} > 0$ and, by Proposition \ref{prop:Vc0}, we have $V=V_c$.
\item If $V_c(0) \leq 0=V_d(0)$, then, by the Comparison Principle in Theorem~\ref{th:comp}, we have $V_c\leq V_d$. Consequently, by the Comparison Principle, we must have $V_d'(0) \leq \beta=V_c'(0)$. Indeed, if we had $V_d'(0)>\beta=V_c'(0)$, then the Comparison Principle would imply that $V_d \leq V_c$. To conclude, let us show that $V(0)=0$, which by Proposition~\ref{prop:Vc00}, will give us $V=V_d$, as expected. To do so, let us proceed by contradiction and assume that $V(0)>0$. In this case, by Proposition~\ref{prop:Vc0}, we have $V=V_c$ and in particular $V(0)=V_c(0) \leq 0$ (where this last inequality is our initial assumption), which is a contradiction.
\item Third, if $V_d'(0) \leq \beta = V_c'(0)$, then, by the Comparison Principle, we have $V_d \geq V_c$ and in particular $V_d(0)=0 \geq V_c(0)$. Consequently, by (2), we have $V=V_d$.
\item Finally, if $V_d'(0) > \beta = V_c'(0)$, then, by the Comparison Principle, we have $V_c \geq V_d$ and in particular $V_c(0) \geq 0=V_d(0)$. If $V_c(0)>0$, then by (1) we have $V=V_c$. If $V_c(0)=0$, then by (2) we have $V=V_d$.
\end{enumerate}

\appendix
\section{Proof of Proposition~\ref{prop.Vcprime}}\label{Appendix:proof:Vcprime}

From Proposition~\ref{prop:immediate-injection}, we know that, for $x \geq 0$,
\[
\limsup_{h \to 0+} \frac{V_c(x+h)-V_c(x)}{h} \leq \beta.
\]
We argue by contradiction and assume that
\[
\liminf_{h \to 0+} \frac{V_c(h)-V_c(0)}{h} < \beta .
\]
In this case, there exists $0<\eta<\beta$ and a positive sequence $h_n  \to 0$ such that
\[
V_c(h_n ) \leq V_c(0) + (\beta - \eta) h_n  \; , \quad \text{for each $n \geq 1$.}
\]
Then, using Proposition~\ref{prop:immediate-injection}, we have that, for $x \geq h_n $,
\begin{equation}\label{eq:contradiction-inequality}
V_c(x) \leq V_c(0) - \eta h_n  + \beta x .
\end{equation}

In the rest of the proof, we will use the following notation: for an integer $n \geq 1$ and an arbitrary $\pi =(\ell,K) \in \Pi_c$, set $\tau^\pi_n := \inf\{t \geq 0 \colon X^\pi_t \geq h_n \}$ and set
\[
\mathcal I^\pi := \int_{[0,\tau^\pi_n]} e^{-q t}(\ell_t \diff t - \beta \diff K_t) .
\]
Note that $\mathcal{I}^\pi$ is well defined on $\BRA{\tau_n^\pi=\infty}$ by the admissibility Condition~\eqref{cond-admissibility-on-K}.

By the Dynamic Programming Principle in~\eqref{DPP}, we can write
\begin{align}
V_c(0) &= \sup_{\pi = (\ell,K) \in \Pi_c} \E_0 \SBRA{\mathcal I^\pi + e^{- q \tau^\pi_n}  V_c(X^\pi_{\tau^\pi_n})\ind_{\tau_n^\pi<\infty}} \notag\\
&\leq \sup_{\pi = (\ell,K) \in \Pi_c} \E_0 \SBRA{\mathcal I^\pi + e^{- q \tau^\pi_n } \left( V_c(0) - \eta h_n  + \beta X^\pi_{\tau^\pi_n} \right)\Ind_{\tau_n^\pi<\infty} }\notag\\
&= V_c(0) - \eta h_n  + \sup_{\pi = (\ell,K) \in \Pi_c} \E_0 \SBRA{\mathcal I^\pi + \left(1-e^{- q \tau^\pi_n }\right) \left( \eta h_n  - V_c(0) \right) + \beta e^{- q \tau^\pi_n } X^\pi_{\tau^\pi_n}\Ind_{\tau_n^\pi<\infty}} .
\label{eq:V_c(0)-1}
\end{align}
where, in the inequality, we used~\eqref{eq:contradiction-inequality}.

Let $m\geq 1$ and denote by $\tau^{\pi,m}_n=\tau^\pi_n\wedge m$ and $\mathcal{I}^{\pi,m}$ the integral $\mathcal{I}^\pi$ on $[0, \tau^{\pi,m}_n]$. Using Itô's Formula applied to $e^{- q \tau^{\pi,m}_n } X^\pi_{\tau^{\pi,m}_n}$ and taking expectations, we can further write
 \begin{align*}
& \E_0 \SBRA{\mathcal I^{\pi,m} + \left(1-e^{- q \tau^{\pi,m}_n }\right) \left( \eta h_n  - V_c(0) \right) + \beta e^{- q \tau^\pi_n } X^\pi_{\tau^{\pi}_n}\Ind_{\tau^\pi_n\leq m}}
 \\
 &\leq \E_0 \SBRA{\mathcal I^{\pi,m} + \left(1-e^{- q \tau^{\pi,m}_n }\right) \left( \eta h_n  - V_c(0) \right) + \beta e^{- q \tau^{\pi,m}_n } X^\pi_{\tau^{\pi,m}_n}}\\
 &=\E_0 \left[ (1-\beta) \int_0^{\tau^{\pi,m}_n} e^{-q t} \ell_t \diff t + \left(1-e^{- q \tau^{\pi,m}_n }\right) \left( \eta h_n  - V_c(0) \right)+\beta \int_0^{\tau^{\pi,m}_n} e^{-q t} \left( \mu(X^\pi_t) - q X^\pi_t \right) \diff t \right]\\
 &\leq  \E_0 \SBRA{\left(1-e^{- q \tau^{\pi,m}_n  }\right) \left( \eta h_n  - V_c(0) \right) + \beta \mu(0) \int_0^{\tau^{\pi,m}_n } e^{-q t} \diff t}\\
 &= q \PAR{\eta h_n  - V_c(0) + \beta\frac{\mu(0)}{q}}\E_0 \SBRA{ \int_0^{\tau^{\pi,m}_n } e^{-q t} \diff t},
 \end{align*}
 where, in the second inequality, we used the fact that $\ell_t \geq 0$ and $\beta > 1$, and we used both elements of $(1)$ in Assumption~\ref{assumption:SDE}. Taking the limit when $m\to \infty$ on both sides of the inequality via the monotone and the dominated convergence theorems, and combining the resulting inequality with \eqref{eq:V_c(0)-1}, we obtain
\begin{align}\label{eq:V_c(0)-2}
V_c(0)
\leq V_c(0) - \eta h_n  + q \left(\eta h_n  - V_c(0) + \beta \frac{\mu(0)}{q} \right) \sup_{\pi = (\ell,K) \in \Pi_c} \E_0 \SBRA{\int_0^{\tau^\pi_n} e^{-q t} \diff t}.
\end{align}

If $V_c(0) - \beta \frac{\mu(0)}{q} > 0$, then there exists an integer $N \geq 1$ such that, for any $n\geq N$,
\[
\eta h_n  - V_c(0) + \beta \frac{\mu(0)}{q} < 0
\]
and thus $V_c(0) \leq V_c(0) - \eta h_n $, which is a contradiction.

Therefore, for the rest of the proof, we assume $V_c(0) - \beta \frac{\mu(0)}{q} \leq 0$. Under this assumption, defining for each $n\geq 1$ the positive constant $A_n = q \left(\eta h_n  - V_c(0) + \beta \frac{\mu(0)}{q} \right)$, we can write by~\eqref{eq:V_c(0)-2}, for each $n \geq 1$,
\begin{equation}\label{eq:V_c(0)-3}
V_c(0) \leq V_c(0) - \eta h_n  + A_n \sup_{\pi = (\ell,K) \in \Pi_c} \E_0 \SBRA{\int_0^{\tau^\pi_n} e^{-q t} \diff t} .
\end{equation}

Now, for $n\geq 0$, define $\varphi(x) = x^2 - h_n^2$. Note that $\varphi$ is nonpositive, increasing and bounded on $[0,h_n]$. For a strategy $(\ell,K) \in \Pi_c$, recall that $\tau^\pi_n = \inf\{t \geq 0 \colon X^\pi_t \geq h_n\}$. By It\^o's Formula and then by taking expectations (replacing $\tau^\pi_n$ with $\tau^\pi_n\wedge m$ and then taking the limit when $m\to \infty$ as above), we obtain
\begin{align*}
0 &\geq \E_0 \SBRA{ e^{-q \tau^\pi_n} \varphi \left( X^\pi_{\tau^\pi_n -} \right)} \\
&= \varphi(0) + \E_0 \SBRA{ \int_0^{\tau^\pi_n} e^{-q t} \left( \sigma^2(X^\pi_t) + 2 X^\pi_t \mu(X^\pi_t) - q \varphi \left( X^\pi_t \right) \right) \diff t} \\
& \qquad \qquad + \E_0 \SBRA{ \int_0^{\tau^\pi_n} e^{-q t} 2 X^\pi_t \left( \diff K^c_t - \ell_t \diff t \right)} + \E_0 \SBRA{\sum_{0 \leq t < \tau^\pi_n} e^{-qt} \PAR{\varphi(X_t)-\varphi(X_{t-})} \Ind_{\Delta K_t >0}} \\
&\geq -h_n^2 + \E_0 \SBRA{ \int_0^{\tau^\pi_n} e^{-q t} \left( \sigma^2(X^\pi_t) + 2 X^\pi_t \mu(X^\pi_t) - q \varphi \left( X^\pi_t \right) \right) \diff t} - \E_0 \SBRA{ \int_0^{\tau^\pi_n} e^{-q t} 2 X^\pi_t \ell_t \diff t} \\
&\geq -h_n^2 + \E_0 \SBRA{ \int_0^{\tau^\pi_n} e^{-q t} \left( \underline{\sigma}^2 - 2 h_n F(h_n) + 2 X^\pi_t \mu(X^\pi_t) - q \varphi \left( X^\pi_t \right) \right) \diff t} ,
\end{align*}
in which we used the fact that the jumps are nonnegative and the fact that, by the properties of $F$ (together with the admissibility of a strategy), we can write
\[
\E_0 \SBRA{ \int_0^{\tau^\pi_n} e^{-q t} X^\pi_t \ell_t \diff t} \leq h_n F(h_n) \E_0 \SBRA{ \int_0^{\tau^\pi_n} e^{-q t} \diff t} .
\]
 For $t \in [0,\tau^\pi_n)$, we can write
\[
\underline{\sigma}^2 - 2 h_n F(h_n) + 2 X^\pi_t \mu(X^\pi_t) - q \varphi \left( X^\pi_t \right) \geq \underline{\sigma}^2+a(n),
\]
with $a(n):=-2 h_n F(h_n) + 2\min_{x\in[0,h_n]} x \mu(x)$.
Finally, we have obtained
\[
h_n^2 \geq \left( \underline{\sigma}^2 + a(n) \right) \E_0 \SBRA{ \int_0^{\tau^\pi_n} e^{-q t} \diff t} ,
\]
with $\lim_{n \to 0} a(n) = 0$. Therefore, there exists $n_0 \geq 1$ (independent of $\pi$) such that, for any $n \geq n_0$, we have $B_n := \underline{\sigma}^2 +a(n)>0$ and hence
\[
\frac{h_n^2}{B_n} \geq \E_0 \SBRA{ \int_0^{\tau^\pi_n} e^{-q t} \diff t} .
\]
Using this last result into \eqref{eq:V_c(0)-3}, we  have, for $n$ sufficiently large,
\begin{align*}
V_c(0) &\leq V_c(0) - \eta h_n  + \frac{A_n}{B_n} h_n ^2 = V_c(0) - h_n  \left( \eta - \frac{A_n}{B_n} h_n  \right) ,
\end{align*}
with $h_n  \left( \eta - \frac{A_n}{B_n} h_n  \right) > 0$. This is a contradiction.

\section{Proof of Lemma~\ref{lemma.super_comp}}\label{App.super_comp}

Fix $x>0$ and fix an arbitrary strategy $\pi = (\ell,K) \in \Pi_c$. For simplicity, in this proof, we write $X=X^\pi$. Define $\tau = \inf\{t \geq 0 \colon X_t = 0\}$. For an integer $n \geq 1$, using Itô's Formula we have
\begin{multline*}
e^{-q (\tau \wedge n)} \psi \PAR{X_{\tau \wedge n}} = \psi \PAR{X_{0 -}} + \int_{0}^{\tau \wedge n} e^{-qt} \PAR{\frac{\sigma^2(X_t)}{2} \psi^{\prime \prime}(X_t) + (\mu(X_t)-\ell_t) \psi^\prime (X_t) - q\psi(X_t)} \diff t \\
+ \int_{0}^{\tau \wedge n} e^{-qt} \sigma(X_t) \psi^\prime (X_t) \diff W_t + \int_{0}^{\tau \wedge n} e^{-qt} \psi^\prime (X_t) \diff K^c_t \\
+ \sum_{0 \leq t \leq \tau \wedge n} e^{-qt} \PAR{\psi(X_t)-\psi(X_{t-})} \Ind_{\Delta K_t >0} .
\end{multline*}
Adding/removing $\int_{0}^{\tau \wedge n} e^{-qt} \ell_t \diff t$ and then using the two inequalities given by~\eqref{eq:HJB-supersol} and the fact that $\psi$ is nonnegative, we can further write
\begin{equation*}
e^{-q \tau} \psi \PAR{X_\tau} \Ind_{\tau<n} \leq  \psi \PAR{X_{0 -}} + \int_{0}^{\tau \wedge n} e^{-qt} \sigma(X_t) \psi^\prime (X_t) \diff W_t - \int_{0}^{\tau \wedge n } e^{-qt} \ell_t \diff t + \beta \int_{[0,\tau \wedge n ]} e^{-qt} \diff K_t .
\end{equation*}
Taking expectations with $X_{0-}=x$  and re-arranging, we get
\begin{align*}
\psi(x) &\geq \E_x\SBRA{ e^{-q \tau} \psi \PAR{X_\tau} \Ind_{\tau< n}} + \E_x\SBRA{ e^{-q n} \psi \PAR{X_n} \Ind_{n < \tau}} \\
& \qquad + \E_x\SBRA{ \int_{0}^{\tau \wedge n } e^{-qt} \ell_t \diff t} - \beta \E_x \SBRA{\int_{[0,\tau \wedge n ]} e^{-qt} \diff K_t} \\
&\geq \E_x\SBRA{ e^{-q \tau} \psi \PAR{0} \Ind_{\tau<n }} + \E_x\SBRA{ e^{-q n} \psi \PAR{X_n} \Ind_{n < \tau}} \\
& \qquad + \E_x\SBRA{ \int_{0}^{\tau \wedge n } e^{-qt} \ell_t \diff t} - \beta \E_x \SBRA{\int_{[0,\tau \wedge n ]} e^{-qt} \diff K_t} .
\end{align*}
Since $X_\tau=0$ and $\psi(0) \geq V_c(0)$, taking limits when $n \to \infty$ (using the monotone convergence theorem for each expectations, thanks to admissibility), we finally get
\[
\psi(x) \geq \E_x\SBRA{ e^{-q \tau} V_c \PAR{0} \Ind_{\tau<\infty}} + \E_x\SBRA{ \int_{0}^{\tau} e^{-qt} \ell_t \diff t - \beta \int_{[0,\tau]} e^{-qt} \diff K_t} .
\]
Finally, since $X_\tau \equiv X^\pi_\tau=0$ and since the strategy $\pi \in \Pi_c$ is arbitrary, we have obtained
\[
\psi(x) \geq \sup_{\pi=(\ell,K) \in \admissible_c} \E_x\SBRA{ \int_{0}^{\tau} e^{-qt} \ell_t \diff t - \beta \int_{[0,\tau]} e^{-qt} \diff K_t + e^{-q \tau} V_c \PAR{X^\pi_\tau} \Ind_{\tau<\infty}} .
\]
The result follows from the Dynamic Programming Principle in~\eqref{DPP}.

\section{Proof of Lemma~\ref{lemma.sub_comp}}\label{App.sub_comp}

As $\phi(x)=A+Bx$ is a subsolution on $(a,b)$ of~\eqref{eq.HJB} with $0<B<\beta$, we have, for all $x\in(a,b)$,
    \[
    q\phi(x)\leq B\mu(x)+(1-B)F(x)\Ind_{B<1}.
    \]

Take $\pi=(\ell,K) \in \Pi_c$ with $\ell_t=F(X_t^\pi)\ind_{B<1}$ and $K_t = 0$ for $t \in [0,\tau]$, where $\tau=\inf\BRA{t\geq 0: X^\pi_t\notin (a,b)}$. Note that, on $\BRA{\tau<\infty}$, if $b<\infty$ then $X^\pi_{\tau}\in\BRA{a,b}$, while if $b=\infty$ then $X^\pi_{\tau}=a$. Then, by assumption, we have $\phi(X^\pi_{\tau})\leq V_c(X^\pi_\tau)$ on $\BRA{\tau<\infty}$.

Using Itô's Formula, taking expectations with $X_{0-}=x \in (a,b)$ and using the above facts, we can write, for $n\geq 1$,  
\begin{align}
        \phi(x) & = \mathbb{E}_x \SBRA{ e^{-q(\tau \wedge n)} \phi(X^\pi_{\tau \wedge n}) + \int_0^{\tau \wedge n} e^{-q t} \left(q \phi(X^\pi_t)-\mu(X^\pi_t)B  + F(X^\pi_t) B \ind_{B< 1} \right) \diff t } \notag\\
        & \leq \mathbb{E}_x \SBRA{ e^{-q\tau} \phi(X^\pi_{\tau}) \ind_{\{\tau \leq n\}} + e^{-q n} \phi(X^\pi_n) \ind_{\{\tau > n\}}  + \int_0^{\tau \wedge n} e^{-q t} F(X^\pi_t) \ind_{B< 1} \diff  t } \notag\\
        & \leq \mathbb{E}_x\SBRA{ e^{-q\tau} V_c(X^\pi_{\tau}) \ind_{\{\tau \leq n\}}   + \int_0^{\tau} e^{-q t} \ell_t  \diff  t } + e^{-q n} \mathbb{E}_x\SBRA{\phi(X^\pi_n) \ind_{\{\tau > n\}}} .
   \label{eq:phi-subsolution}
    \end{align}
We want to show that the result follows by taking the limit in~\eqref{eq:phi-subsolution} and then invoking the Dynamic Programming Principle.

First, by the monotone convergence theorem, we have
\[
\lim_{n\to \infty} \mathbb{E}_x \SBRA{e^{-q\tau} V_c(X^\pi_{\tau}) \ind_{\{\tau \leq n\}}} = \mathbb{E}_x \SBRA{e^{-q\tau} V_c(X^\pi_{\tau}) \ind_{\{\tau <\infty\}}} .
\]
Consequently, all is left to show is that $\displaystyle{\lim_{n\to \infty} e^{-q n} \mathbb{E}_x \SBRA{ \phi(X^\pi_n) \ind_{\{\tau > n\}}}=0}$. However, since $\phi$ is an affine function, this follows from the transversality property established in Lemma~A.1 of \cite{guerin-et-al_2024}. Although the proof in \cite{guerin-et-al_2024} uses the assumption $\mu(0)>0$, removing this assumption only affects the value of a constant, so the conclusion remains valid.

\section{Proof of Theorem~\ref{th:comp}}\label{sec:comparisonC2}

Let $\phi, \psi \in \mathcal W$ be such that $\phi$ (resp. $\psi$) is a subsolution (resp.\ supersolution) of the HJB equation~\eqref{eq.HJB_d}, with either $\phi(0) \leq \psi(0)$ or $\phi'(0) \geq \psi'(0)$.

Recall that $\mu'(0) < q$. Let $q_0>0$ and $\alpha_0>1$ such that $\alpha_0 \mu'(0) < q_0 < q.$ Since $\mu$ is concave, there exists $x_0>0$ such that $\alpha_0 \frac{\mu(x)}{x+x_0}<q_0$ for any $x\geq 0$. Furthermore, $$\varsigma : = \sup_{x\geq 0} \frac{\sigma^2(x) }{(x+x_0)^2} < \infty,$$ by assumption. Let $1 < \alpha < \min(2,1+\frac{q-q_0}{\varsigma},\alpha_0)$ and introduce $\psi_2(x) = A_0 + A_1 x + (x+x_0)^\alpha,$ for some $A_0 > \phi(0)$ and $A_1 < \phi'(0)-\alpha x_0^{\alpha-1}.$ 

It is straightforward, though somewhat technical, to check that there exists $A_0$ large enough such that $\psi_2$ is a strict supersolution of \eqref{eq.HJB_d}.

If we let $\lambda \in (0,1),$ and define $U = \phi -  \lambda \psi - (1-\lambda) \psi_2$. To prove that $U \leq 0$, we construct a contradiction by assuming that $$\sup_{x\geq 0} U(x) > 0.$$ Since  $\phi \in \mathcal W$ and $\alpha>1$, we know that $\lim_{x\to \infty} U(x) = -\infty$, which means that the supremum is attained at some $\hat x \geq 0.$ However, since $\psi_2(0) > \phi(0)$ and $\psi'_2(0) < \phi'(0),$ we readily see that $U(0) < 0$ or $U'(0) > 0$, depending on whether $\phi(0) \leq \psi(0)$ or $\phi'(0) \geq \psi'(0)$. In either case, we deduce  $\hat x > 0.$

Because the supremum of $U$ is attained at $\hat x,$ 
it must be the case that $U'(\hat x) = 0$ and 
$U''(\hat x) \leq 0.$ According to the sub/supersolution properties of $\phi,\psi $ and $\psi_2$, it follows that
\begin{multline}\label{eq:ineq U}
q U(\hat x) - \frac{1}{2} \sigma(\hat x)^2 U''(\hat x)- F(\hat x)\Big( (1-\phi'(\hat x)) \ind_{\phi'(\hat x)<1} -\lambda (1-\psi'(\hat x))\ind_{\psi'(\hat x)<1} \\
-(1-\lambda)(1-\psi_2'(\hat x))\ind_{\psi_2'(\hat x)<1}\Big) <0.
\end{multline}
When $\phi'(\hat x)\geq 1$, we easily deduce that $U(\hat x)<0$. When $\phi'(\hat x)<1$, at least $\psi'(\hat x)<1$ or $\psi_2'(\hat x)<1$ since $\phi'(\hat x)=\lambda\psi'(\hat x)+(1-\lambda)\psi_2'(\hat x)$ (recall that $U'(\hat x)=0$). Consequently, we also deduce $U(\hat x)<0$ from \eqref{eq:ineq U}.
We then have a contradiction. Since $\lambda$ is arbitrary, we conclude that $\phi \leq \psi$ on $[0,\infty).$

\bibliographystyle{alpha}
\bibliography{references.bib}
\end{document}